\documentclass[letterpaper]{IEEEtran}
\usepackage{amsmath}
\usepackage{amsfonts}
\usepackage{comment}
\usepackage{amssymb}
\usepackage{float}
\usepackage{wrapfig}
\usepackage{xcolor}
\usepackage{enumerate}
\restylefloat{table}
\def\OPT{{\rm OPT}}

\def\LP{{\rm LP}}

\def\DP{{\rm DP}}

\def\Cl{\underline{C}}
\def\Cu{\overline{C}}

\def\err{{\rm err}}
\def\com{{\rm com}}
\def\uP{\overline{{\rm P}}(n;R)}

\def\lP{\underline{{\rm P}}(n;R)}
\def\mix{{\rm mix}}

\def \Xscrt{\widetilde{\mathcal{X}}}

\def \Vscr{\mathcal{V}}

\def\talpha {\widetilde{\alpha}}

\def \alphah{\widehat{\alpha}}
\def \deltah{\widehat{\delta}}
\def \Thetah{\widehat{\Theta}}
\def \Xscrh{\widehat{\mathcal{X}}}
\def \Xscrt{\widetilde{\mathcal{X}}}
\def \Xh{\widehat{X}}
\def \Xt{\widetilde{X}}
\def \Yscrh{\widehat{\mathcal{Y}}}
\def \Yscrt{\widetilde{\mathcal{Y}}}
\def \Yh{\widehat{Y}}
\def \Yt{\widetilde{Y}}

\def \ugame{\overline{\vartheta}(n;R)}

\def \lgame{\underline{\vartheta}(n;R)}
\def \Uo{\overline{\Upsilon}(R)}
\def \Uu{\underline{\Upsilon}(R)}

\def\csf{\mathsf{c}}
\def\ssf{\mathsf{s}}
\usepackage{amsmath, amssymb, bbm, xspace}
\usepackage{epsfig}
\usepackage{longtable}
\usepackage{color}
\usepackage{mathrsfs}
\usepackage{comment}

\usepackage{courier}
\newtheorem{theorem}{Theorem}[section]

\newtheorem{lemma}[theorem]{Lemma}

\newtheorem{proposition}[theorem]{Proposition}

\newtheorem{corollary}[theorem]{Corollary}

\def\bkE{{\rm I\kern-.17em E}}
\def\bk1{{\rm 1\kern-.17em l}}
\def\bkD{{\rm I\kern-.17em D}}
\def\bkR{{\rm I\kern-.17em R}}
\def\bkP{{\rm I\kern-.17em P}}

\def\bkZ{{\bf{Z}}}

\def\bkE{{\rm I\kern-.17em E}}
\def\bk1{{\rm 1\kern-.17em l}}
\def\bkD{{\rm I\kern-.17em D}}
\def\bkR{{\rm I\kern-.17em R}}
\def\bkP{{\rm I\kern-.17em P}}

\makeatletter
\newcommand{\pushright}[1]{\ifmeasuring@#1\else\omit\hfill$\displaystyle#1$\fi\ignorespaces}
\newcommand{\pushleft}[1]{\ifmeasuring@#1\else\omit$\displaystyle#1$\hfill\fi\ignorespaces}
\makeatother

\def\bkZ{{\bf{Z}}}
\def\b12{(\beta_1,\beta_2)}

\newenvironment{proofarg}[1][]{\noindent\hspace{2em}{\itshape Proof #1: }}{\hspace*{\fill}~\qed\par\endtrivlist\unskip}
\newenvironment{example}{{\noindent \bf Example}}{\hfill $\square$\hspace{-4.5pt}\vspace{6pt}}
\newcounter{example}
\renewcommand{\theexample}{\thesection.\arabic{example}}

\newcounter{remark}
\renewcommand{\theremark}{\thesection.\arabic{remark}}
\newenvironment{remarkc}[1][]{\refstepcounter{remark}
\noindent{\itshape Remark~\theremark. #1} \rmfamily}{\hspace*{\fill}~$\square$\vspace{0pt}}

\def\t{^\top}
\def\Bscr{\mathscr{B}}
\def\Xscr{\mathcal{X}}
\def\Yscr{\mathcal{Y}}

\def\Ebb{\mathbb{E}}
\newlength{\noteWidth}
\long\def\notes#1{\ifinner
{\tiny #1}
\else
\marginpar{\parbox[t]{\noteWidth}{\raggedright\tiny #1}}
\fi\typeout{#1}}

 \def\notes#1{\typeout{read notes: #1}} %uncomment for final version

\newcommand{\I}[1]{\mathbb{I}_{\{#1\}}}

\newcommand{\ie}{i.e.\@\xspace} %%% i.e.,
\newcommand{\eg}{e.g.\@\xspace} %%% e.g.,
\newcommand{\etal}{et al.\@\xspace} %%% e.g., Gill \etal (1986)

\newcommand{\Real}{\ensuremath{\mathbb{R}}}

\newcommand{\minimize}[1]{\displaystyle\minim_{#1}}
\newcommand{\minim}{\mathop{\hbox{\rm min}}}
\newcommand{\maximize}[1]{\displaystyle\maxim_{#1}}
\newcommand{\maxim}{\mathop{\hbox{\rm max}}}

\def\OPT{{\rm OPT}}

\def\Ebb{\mathbb{E}}

\def\Pbb{{\mathbb{P}}}
\def\Nbb{{\mathbb{N}}}

\def\Xbb{{\mathbb{X}}}
\def\Ybb{{\mathbb{Y}}}
\def\Ibb{{\mathbb{I}}}

\def\limn{\ds \lim_{n \rightarrow \infty}}

\def\exp{\mathop{\hbox{\rm exp}}}

\def\half  {{\textstyle{1\over 2}}}

\def\spose#1{\hbox to 0pt{#1\hss}}
\def\sub#1{^{\null}_{#1}}
\def\text #1{\hbox{\quad#1\quad}}

\def\fhat{{\hat f}}
\def\xhat{{\hat x}}

\def\nthinsp{\mskip -2   mu}

\def\superstar{^{\raise 0.5pt\hbox{$\nthinsp *$}}}
\def\SUPERSTAR{^{\raise 0.5pt\hbox{$*$}}}
\def\lamstarT {\lambda^{\raise 0.5pt\hbox{$\nthinsp *$}T}}

\def\Ascr{{\cal A}}
\def\Bscr{{\cal B}}

\def\Dscr{{\cal D}}

\def\Tscr{{\cal T}}

\def\Pscr{{\cal P}}

\def\Sscr{{\cal S}}
\def\Sscrhat{\widehat{\mathcal{S}}}

\def\Vscr{{\cal V}}

\def\Zscr{{\cal Z}}
\def\Xscr{{\cal X}}
\def\Yscr{{\cal Y}}

\def\fhat{\widehat f}

\def\ghat{\skew{4.3}\widehat g}
\def\gtilde{\skew{4.5}\widetilde g}

\def\Pbar{\skew5\bar P}

\def\shat{\widehat s}

\def\Shat{\widehat S}

\def\xhat{\skew{2.8}\widehat x}
\def\xtilde{\skew3\widetilde x}

\def\yhat{\skew3\widehat y}
\def\ytilde{\skew3\widetilde y}
\def\Ybar{\skew2\bar Y}
\def\Yhat{\widehat Y}

\def\aur{\;\textrm{and}\;}
\def\where{\;\textrm{where}\;}

\def\non{\nonumber}

\let\forallnew\forall
\renewcommand{\forall}{\forallnew\ }
\let\forall\forallnew

\def\ds{\displaystyle}
		\def\bkE{{\rm I\kern-.17em E}}
		\def\bk1{{\rm 1\kern-.17em l}}
		\def\bkD{{\rm I\kern-.17em D}}
		\def\bkR{{\rm I\kern-.17em R}}
		\def\bkP{{\rm I\kern-.17em P}}
		\def\bkY{{\bf \kern-.17em Y}}
		\def\bkZ{{\bf \kern-.17em Z}}
		\def\bkC{{\bf  \kern-.17em C}}
{\begin{list}{}%
         {\setlength{\leftmargin}{#1}}%
         \item[]%
}
{\end{list}}

		\def\bsp{\begin{split}}
		\def\beq{\begin{eqnarray}}
		\def\bal{\begin{align*}}
		\def\bc{\begin{center}}
		\def\be{\begin{enumerate}}
		\def\bi{\begin{itemize}}
		\def\bs{\begin{small}}
		\def\bS{\begin{slide}}
		\def\ec{\end{center}}
		\def\ee{\end{enumerate}}
		\def\ei{\end{itemize}}
		\def\es{\end{small}}
		\def\eS{\end{slide}}
		\def\eeq{\end{eqnarray}}
		\def\eal{\end{align*}}
		\def\esp{\end{split}}
		\def\qed{ \vrule height7.5pt width7.5pt depth0pt}  %width4.17pt depth0pt} 

	\def\maxproblemsmall#1#2#3#4{\fbox
		 {\begin{tabular*}{0.47\textwidth}
			{@{}l@{\extracolsep{\fill}}l@{\extracolsep{6pt}}l@{\extracolsep{\fill}}c@{}}
				#1 & $\maximize{#2}$ & $#3$ & $ $ \\[5pt]
					 & $\subject\ $    & $#4$ & $ $
			\end{tabular*}}
			}

	\def\cp2problem#1#2#3#4{\fbox
		 {\begin{tabular*}{0.9\textwidth}
			{@{}l@{\extracolsep{\fill}}l@{\extracolsep{6pt}}l@{\extracolsep{\fill}}c@{}}
				#1 & & $#4 $ 
			\end{tabular*}}}

		\def\bkE{{\rm I\kern-.17em E}}
		\def\bk1{{\rm 1\kern-.17em l}}
		\def\bkD{{\rm I\kern-.17em D}}
		\def\bkR{{\rm I\kern-.17em R}}
		\def\bkP{{\rm I\kern-.17em P}}
		
		\def\bkZ{{\bf{Z}}}

\newcommand {\beeq}[1]{\begin{equation}\label{#1}}
\newcommand {\eeeq}{\end{equation}}
\newcommand {\bea}{\begin{eqnarray}}
\newcommand {\eea}{\end{eqnarray}}

\def\texitem#1{\par\smallskip\noindent\hangindent 25pt
               \hbox to 25pt {\hss #1 ~}\ignorespaces}

\def\bsp{\begin{split}}
		\def\beq{\begin{eqnarray}}
		\def\bal{\begin{align*}}
		\def\bc{\begin{center}}
		\def\be{\begin{enumerate}}
		\def\bi{\begin{itemize}}
		\def\bs{\begin{small}}
		\def\bS{\begin{slide}}
		\def\ec{\end{center}}
		\def\ee{\end{enumerate}}
		\def\ei{\end{itemize}}
		\def\es{\end{small}}
		\def\eS{\end{slide}}
		\def\eeq{\end{eqnarray}}
		\def\eal{\end{align*}}
		\def\esp{\end{split}}
		\def\qed{ \vrule height7.5pt width7.5pt depth0pt}  %width4.17pt depth0pt} 

                        \newenvironment{examplee}{{ \emph{Example:} }}{\hfill $\Box$ \vspace{6pt}}

\usepackage{amsmath, amssymb, xspace}
\usepackage{epsfig}
\usepackage{longtable}
\usepackage{color}
\usepackage{mathrsfs}
\usepackage{subfig}
\newenvironment{proof}[1][]{{\noindent \emph {Proof} #1: }}{\hfill \qed \vspace{3pt}\\ }

\def\sub{\hbox{\rm s.t}}

			\def\problemsmalla#1#2#3#4{\fbox
		 {\begin{tabular*}{0.47\textwidth}
			{@{}l@{\extracolsep{\fill}}l@{\extracolsep{-4pt}}l@{\extracolsep{\fill}}c@{}}
				#1 &  & $\minimize{#2}$  $#3$ & $ $ \\[4pt]
					  $\sub \ $  &   & $#4$ &  $ $
			\end{tabular*}}
			}

	\def\maxproblemsmall#1#2#3#4{\fbox
		 {\begin{tabular*}{0.47\textwidth}
			{@{}l@{\extracolsep{\fill}}l@{\extracolsep{-4pt}}l@{\extracolsep{\fill}}c@{}}
				#1 &  & $\maximize {#2}$ $#3$ & $ $ \\[4pt]
					 $\sub \ $  &   & $#4$ & $ $
			\end{tabular*}}
			}

			\renewcommand{\I}[1]{\mathbb{I}\{#1\}}
\ifCLASSINFOpdf
  \else
 
\fi

\author{\IEEEauthorblockN{Sharu Theresa Jose \quad Ankur A. Kulkarni\thanks{Ankur with the Systems and Control Engineering group, 
Indian Institute of Technology Bombay,
Mumbai, 400076, India. 
 This work was done partly while Sharu was a Ph.D. student in this group. They can be contacted at sharutheresa@iitb.ac.in, kulkarni.ankur@iitb.ac.in. This work will be presented in part at the IEEE Conference on Decision and Control, to be held in December 2018~\cite{jose2018game}.}}}
\title{Shannon meets von Neumann: A Minimax Theorem for Channel Coding in the Presence of a Jammer}
\begin{document}
\maketitle
\begin{abstract}
%We considered a game between a team comprising of a finite blocklength encoder and decoder and a finite state jammer where the former team attempts to minimize the probability of error and the jammer attempts to maximize it. The nonclassicality of the information structure renders the team's decision problem nonconvex whereby there may not exist a saddle point value to this game. Despite this, we showed that for all but finitely many rates, an asymptotic saddle-point value exists for this game and derived an exact characterization of this value. Our results demonstrate a deeper relation between compound and mixed channels and provide an new characterization of the $\epsilon$-capacity of a compound channel under stochastic codes.
%We consider the setting of channel coding in the finite blocklength regime over a family of channels whose state is controlled by an adversarial jammer. The solutions traditionally sought by information theorists have concentrated on what the communication system can achieve in the worst case over all possible states; we call this the Shannon solution. However, one could also study this problem from the vantage point of a game between the communication system and the jammer and employ solution concepts from game theory to ‘solve’ the game and analyze the system; we refer to the resulting solution as the von Neumann solution. While a Shannon solution always exists, we find that a von Neumann solution need not due inherent nonconvexities in the problem. 
We study the setting of  channel coding over a family of channels whose state is controlled by an adversarial jammer by  viewing it as a zero-sum game between a finite blocklength encoder-decoder team, and the jammer. The encoder-decoder team choose stochastic encoding and decoding strategies to minimize the average probability of error in transmission, while the jammer chooses a distribution on the state-space to maximize this probability. The min-max value of this game is equivalent to channel coding for a compound channel -- we call this the Shannon solution of the problem. The max-min value corresponds to finding a mixed channel with the largest value of the minimum achievable probability of error. When the min-max and max-min values are equal, the problem is said to admit a saddle-point or von Neumann solution. 
While a Shannon solution always exists, a von Neumann solution need not, owing to inherent nonconvexity in the communicating team's problem. Despite this, we show that the min-max and max-min values become equal asymptotically in the large blocklength limit, for all but finitely many rates. We explicitly characterize this limiting value as a function of the rate and obtain tight finite blocklength bounds on the min-max and max-min value. As a corollary we get an explicit expression for the $\epsilon$-capacity of a compound channel under stochastic codes -- the first such result, to the best of our knowledge.  Our results demonstrate a deeper relation between the compound channel and mixed channel than was previously known. They also show that the conventional information-theoretic viewpoint, articulated via the Shannon solution,  coincides  asymptotically with the game-theoretic one articulated via the von Neumann solution. 
 Key to our results is the derivation of new finite blocklength upper bounds on the min-max value of the game via a novel achievability scheme, and lower bounds on the max-min value obtained via the linear programming relaxation based approach we introduced in \cite{jose2016linearIT}.

% a saddle-point solution indeed exists for the above zero-sum game in the limit of large blocklengths for all but finitely many rates. 
% Moreover, we derive an explicit characterization of this limiting saddle-point value which equals the asymptotic minimum probability of error achieved by the stochastic team operating at $R$ in the worst-case over all jammer strategies. This in turn implies that the von Neumann solution is in harmony with the classical Shannon solution, in the limit of large blocklengths. 
% - for rates $R<\Cl$ and $R>\Cu$, where $\Cl,\Cu$ are precomputable thresholds depending on the channel kernels, employing Feinstein's lemma, and for $\Cl<R<\Cu$, a novel split-achievability scheme employing stochastic code- gaurentee tight bounds. 
\end{abstract}
\section{Introduction}
Communication theory has traditionally studied the problem of communication in the presence of a jammer, only from the communicater's perspective. The solutions sought have concentrated on what the \textit{communication system} can achieve in the worst case over all possible scenarios; we call this the \textit{Shannon solution}. 
%In the channel coding setting, this amounts to asking for communication capacity characterizations under some fixed model of adversarial action. 
However, alongside one could also ask a dual question, namely, what is the maximum damage that the \textit{jammer} could achieve in the presence of an intelligent communication system?
%Specifically, could view the problem instead from vantage point of an observer of the \textit{competition} between the communication system and the jammer, in which case one is  compelled to also consider, in equal measure. 
A natural vantage point for analyzing the system from both points of view together is that of a \textit{game} between the communication system and the jammer. One could  then employ solution concepts from  game theory to `solve' the game and analyze the system; we refer to the resulting solution as the \textit{von Neumann solution}. 
The focus of this paper  is a game of the above kind in the finite blocklength regime and the relation between the Shannon and von Neumann solutions of  communication problems in the presence of adversarial action.

We define the following zero-sum game between a finite blocklength communication system and a finite state jammer. The communication system attempts to communicate a discrete uniform source across a discrete memoryless channel by making $n$ uses of the channel. The channel law is determined by its state  which is controlled by the jammer. The communication system comprises of a team of two decision makers -- an encoder and a decoder -- that together attempt to \textit{minimize} the average probability of error incurred in the transmission of the source.  
\begin{wrapfigure}{r}{.30\textwidth}
\includegraphics[scale=0.30,clip=true, trim=.5in 6.1in 0.5in 2.8in]{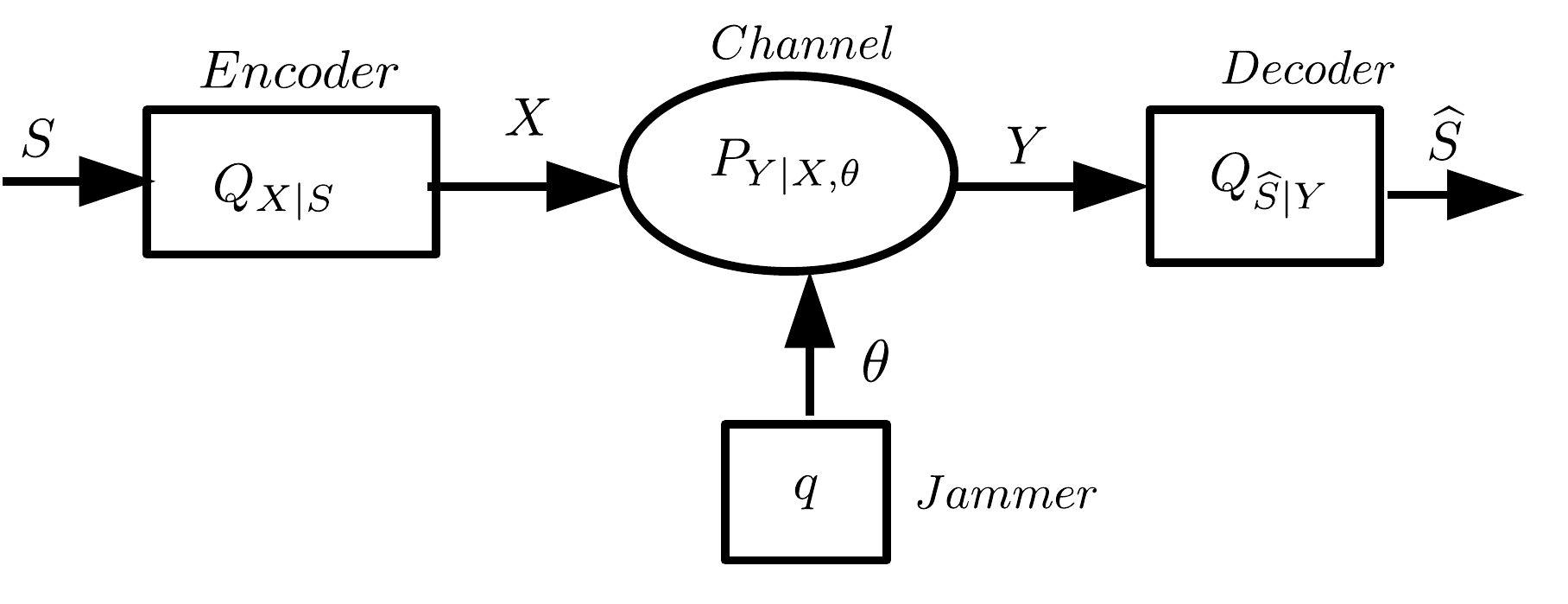}
\caption{A communication system with a finite state jammer}\label{fig:jammer} \vspace{-.5cm}
\end{wrapfigure}
The jammer chooses a state randomly so as to \textit{maximize} the error incurred in transmission. 

Following Fig~\ref{fig:jammer}, let $S$ denote a source message of rate $R$ drawn uniformly at random from $\Sscr:=\{1,\hdots , 2^{nR}\}$. Suppose the jammer chooses a channel state $\theta $ from a set $\Theta$ randomly according to a distribution $q$.  
The encoder and decoder choose \textit{stochastic} strategies. A stochastic encoder maps the source $S$ to an $n$-length channel input $X$ randomly according to a conditional distribution $Q_{X|S}$. 
The channel, in state $\theta$, outputs an $n$-length string denoted $Y$, which is decoded by a stochastic decoder that outputs $\Shat \in \Sscr$ according to a conditional distribution $Q_{\Shat|Y}$. 
We assume that the set of states, $\Theta$, is finite, and that the state once chosen remains fixed through the $n$ uses of the channel. 
The choice of the channel state by the jammer is  independent of the source message and the encoded message, and the actual state realized is not known to either encoder or decoder.

The min-max value or \textit{upper value} of this game is denoted $\ugame$ and is the optimal value of the following problem,
$$\problemsmalla{$\uP$\quad}
	{Q_{X|S},Q_{\Shat|Y}}
	{\displaystyle \max_{q}\quad \Ebb[\Ibb\{S \neq \Shat\}]}
				 {\begin{array}{r@{\ }c@{\ }l}
	Q_{X|S} \in \Pscr(\Xscr|\Sscr), Q_{\Shat|Y} \in \Pscr(\Sscrhat|\Yscr), q \in \Pscr(\Theta).
	\end{array}}$$
where $\Xscr,\Yscr$ are the sets of $n$-length channel inputs and outputs, respectively and $\Pscr(\cdot) $ is the set of probability distributions on `$\cdot$'. The expectation $\Ebb$ above is under the distribution induced by the code $(Q_{X|S},Q_{\Shat|Y})$ and the distribution $q$. 
	Problem $\uP$ thus corresponds to the minimum probability of error achievable by the encoder-decoder team in the worst case over all distributions $q$ chosen by the jammer.
The max-min value or \textit{lower value} of the game, denoted $\lgame$, is  the optimal value of the following problem,
	$$\maxproblemsmall{$\lP$\quad}
	{q}
	{\displaystyle \min_{Q_{X|S},Q_{\Shat|Y}}\quad \Ebb[\Ibb\{S \neq \Shat\}]}
				 {\begin{array}{r@{\ }c@{\ }l}
	Q_{X|S} \in \Pscr(\Xscr|\Sscr), Q_{\Shat|Y} \in \Pscr(\Sscrhat|\Yscr), q \in \Pscr(\Theta),
	\end{array}}$$
	and corresponds to the maximum probability of error achievable by the jammer in the worst case over all stochastic codes employed by the encoder-decoder team.
Note that the following relation always holds:
\begin{align}\bar{\vartheta}(n;R)\geq \underline{\vartheta}(n;R), \qquad \forall n\in \Nbb, R \in [0,\infty). \label{eq:minmax}\end{align} The above zero-sum game  is said to admit a \textit{saddle point} value, or a \textit{von Neumann solution}, if equality holds, \ie, if $\bar{\vartheta}(n;R)=\underline{\vartheta}(n;R). $

%A communication system comprising of an encoder and a decoder wants to communicate a random source $S$  uniformly with distribution $P_S(s)\equiv\frac{1}{M}$,   over $n$ uses of the channel chosen by the jammer. 
%The channel outputs an $n$-length string $Y \in \Yscr$ according to~\eqref{eq:DMC}, which is then decoded according to a randomized mapping $Q_{\Shat|Y}\in \Pscr(\Sscrhat|\Yscr)$ to recover the message $\Shat \in \Sscrhat:=\{1,\hdots,M\}$. The rate of transmission is $R:=\frac{1}{n}\log M$. Note that although the set $\Theta$ of channel kernels is known to both encoder and the decoder, the exact state of the channel during the transmission of $X$ is unknown to both. Moreover, the jammer does not know the input $X$ sent to the channel or the message $S$. The randomized mappings $(Q_{X|S},Q_{\Shat|Y})$ together constitute a \textit{stochastic code}~\cite{csiszar2011information}. 
%

Having defined the game, we ask our first question. Does the game admit a saddle point? For each strategy of the jammer, the team problem of the encoder and decoder has \textit{nonclassical information structure}~\cite{witsenhausen_counterexample_1968}. As argued in~\cite{kulkarni2014optimizer}, in the space of stochastic codes, the communicating team's problem is nonconvex for each jammer strategy. %consequently a minimax theorem is not immediate.
%In his seminal paper,  Witsenhausen~\cite{witsenhausen_counterexample_1968} showed that this information structure can render stochastic team problems nonconvex in the deterministic strategy space. 
%In our model we allow the encoder and decoder to randomize locally (mot, 
This lack of convexity implies that the existence of a saddle point  is not guaranteed. 
Moreover, using the codes that form a saddle point (assuming one exists) may imply capacities, error exponents and suchlike that are distinct from those obtained from the Shannon solution. This leads us to our second question: do the answers obtained from the saddle point, \ie, the von Neumann solution, coincide with those obtained from the Shannon solution? 

We find that, despite the skepticism voiced above, answers to both these questions are in the affirmative in the large blocklength limit, for all but finite many values of the rate $R$ of the source. The main result in this paper shows that  %asymptotically as blocklength $n\rightarrow \infty$, 
 the asymptotic value of the minimum probability of error achievable by a code of rate $R$, in the worst case over all jammer strategies, \textit{equals} the asymptotic value of the maximum probability of error that a jammer can induce, in the worst case over all possible codes of rate of $R$, for all rates $R$ barring some finitely many specific values. In other words, 
  \[\limn \ugame = \limn \lgame, \qquad \forall R \in [0,\infty) \backslash \mathscr{D},\] 
  where $\mathscr{D}$ is a finite set. 
Moreover, $\vartheta(R)$, which is defined as the value of the above limits when they are equal is given by 
\[\vartheta(R) = \max_{q \in \Pscr(\Theta)} \min_{\Theta'\subseteq \Theta} \biggl \lbrace 1 - \sum_{\theta \in \Theta'} q(\theta)| R<C(\Theta') \biggr\rbrace,\]
where $C(\Theta')$ is the capacity of the compound channel formed by $\Theta' \subseteq \Theta$, 
and $\mathscr{D}$ is precisely the set of points of discontinuity of the right-hand side above when viewed as a function of $R \in [0,\infty).$
 At a finite blocklength an \textit{approximate} minimax theorem holds, \ie, the difference $\ugame - \lgame$ becomes vanishingly small with $n$ for all $R\geq 0$ except in $\mathscr{D}$. Interestingly, the upper value, $\ugame$ also happens to equal the probability of error guaranteed by the Shannon solution. 
 Consequently, we find that as the blocklength becomes large, the Shannon solution comes in harmony with the von Neumann solution. As a corollary of these results we also obtain the $\epsilon$-capacity of a compound channel under stochastic codes, the first such result, to the best of our knowledge.

%In the event 
%Hence, a natural question is: does this game admit a saddle point? 
%brings face to face the perspectives of Shannon and von Neumann. From the Shannon perspective one would want to characterize probabilities of error that are achievable by the communication system in the worst case over all possible actions of the jammer; thus this involves  a calculation similar to that of \eqref{eq:sp1}. From the von Neumann perspective, one would want to \textit{also} characterize the probabilities of error that the jammer is able to induce in the worst case over all possible actions of the encoder-decoder team (\ie, \eqref{eq:sp2}), \textit{and then ask}, if the strategies so obtained form a saddle point, \ie, are the resulting probabilities equal? 
%
%
%If a saddle point $(x^*,y^*)$ exists, then $x^*,y^*$ also satisfy \eqref{eq:sp1}-\eqref{eq:sp2} below, 

%Recall that we allow the encoder-decoder pair to randomize their actions locally, \ie, the encoder chooses a distribution on the set of channel inputs for each source message, and the decoder chooses a distribution on the set of destination messages for each channel output.  Thus they employ finite blocklength \textit{stochastic codes} in the sense of Csizsar and Korner \cite{csiszar2004information} or \textit{behavioral strategies} in the language of game theory~\cite{maschler2013game}. The jammer picks a distribution on the state space. 
Despite the nonconvexity of the communicating team's problem argued in~\cite{kulkarni2014optimizer}, our recent results~\cite{jose2016linearIT,jose2017linearitw,jose2018linearAMAC} and \cite{jose2018improvedSW} have shown that all point-to-point problems and several network problems (without a jammer) admit a near-convexity, in the sense that they can be approximated asymptotically by a linear program. This leads us to conjecture that an `approximate' minimax theorem may be within reach, whose approximation becomes increasingly accurate as $n \rightarrow \infty.$ Our results show that this intuition is correct for almost all rates.

%Key to our result is that we have considered the \textit{average} probability of error criterion and the randomization we have allowed.  
 Since the state once chosen is held fixed throughout the $n$ transmissions, problem $\uP$ corresponds to coding for a finite blocklength \textit{compound channel} \cite{blackwell1959capacity} under stochastic codes. Problem $\lP$ on the other hand amounts to foisting a distribution on the state space so that the resulting finite blocklength \textit{mixed} channel \cite{ahlswede1968weak} has the largest probability of error. 
To the best of our knowledge, ours is the first result showing that these are asymptotically the same. The compound and mixed channels have been known to be intimately related (see \cite{blackwell1959capacity,ahlswede1968weak}); in particular they have the same capacity. Our results show that there is an even deeper relation between them.

Interestingly, this minimax theorem breaks down if one considers the maximum probability of error criterion; there is always an interval of rates where there is a gap between the upper and lower values. Furthermore, to the best of our understanding, local randomization produced by stochastic codes seems essential and we have not been able to show similar results with deterministic codes. Local randomization by the encoder allows the code to randomly communicate with the correct codebook with a positive probability, regardless of the action of the jammer. Correctly choosing this randomization achieves the optimal performance. It is unclear if the same can be achieved with deterministic codes. Indeed,  formulating the problem in the space of stochastic codes was also key to our near-convexity results for coding problems found in~\cite{jose2016linearIT,jose2017linearitw,jose2018linearAMAC} and \cite{jose2018improvedSW}. 
%\textbf{Moreover, without the  randomizations we have allowed, we have examples ..}

%Our central contribution lies in showing that for values of all $R$, except for most finitely many specific values, the upper and the lower value of the game coincide in the limit as the blocklength $n\rightarrow \infty$. Thus the finite blocklength game above admits an \textit{approximate} saddle point for almost all values of $R$, and this approximation improves as the blocklength becomes large. 

There are two regimes where our result is rather easy to claim. Define,
\begin{equation}
\Cl=\max_{P_\Xbb} \min_{\theta \in \Theta}I_{P_{\Xbb}}(\Xbb;\Ybb|\theta), \ \ \Cu=\min_{\theta \in \Theta} \max_{P_{\Xbb}} I_{P_{\Xbb}}(\Xbb;\Ybb|\theta), \label{eq:clcu} 
\end{equation}
where $I_{P_{\Xbb}}(\Xbb;\Ybb|\theta)$ is the mutual information between the (single letter) channel input $\Xbb$ and channel output $\Ybb$ when the state is $\theta$ and the channel input distribution is $P_{\Xbb}$. $\Cl$ is the capacity of the compound channel (under deterministic and stochastic codes). Clearly, by definition of capacity, for $R<\Cl$, the upper value and hence the lower value of the game must go to zero as $n\rightarrow \infty$. $\Cu$ is the smallest of the capacities of individual DMCs defined by the state of the channel. For $R>\Cu$, the jammer can choose the state with smallest capacity (with probability one), whereby by the strong converse for channel coding for this DMC, the lower,  and hence upper values  must both approach unity. Thus the asymptotic  minimax theorem holds for $R< \Cl$ and $R>\Cu$. The nontrivial case is $\Cl<R<\Cu$ where we also show the minimax theorem with a new argument that exploits the stochastic encoding allowed.

The limiting saddle point value $\vartheta(R)$ happens to take a peculiar form. It is a step function wherein the points where jumps occur depend on the capacities of compound channels that correspond to subsets of $\Theta$. For example, in the case where $|\Theta|=2$, we have
\[\vartheta(R) =\begin{cases}
0 & \mbox{if} \hspace{0.1cm} R< \Cl, \\
\half & \mbox{if} \hspace{0.1cm}\Cl< R <\Cu, \\
1 & \mbox{if} \hspace{0.1cm} R > \Cu.
\end{cases}
\]
For $|\Theta|=3$ (say $\Theta=\{1,2,3\}$), assuming that the capacities of compound channels corresponding to subsets $\Theta$ satisfy, 
\begin{align*}
\Cl &<C(\{1,2\})<C(\{1,3\})<C(\{2,3\})<\overline{C}=C(\{1\}),
\end{align*} we get
\begin{align*}
\vartheta(R)=\begin{cases} 0 &\mbox{if} \hspace{0.1cm}  R < \Cl ,\\
 \frac{1}{3} &\mbox{if} \hspace{0.1cm}\Cl < R<C(\{1,2\}),\\
\frac{1}{2} &\mbox{if}\hspace{0.1cm} C(\{1,2\})< R\leq C(\{1,3\}),\\
\frac{1}{2} &\mbox{if} \hspace{0.1cm}C(\{1,3\}) < R<C(\{2,3\}),\\
\frac{2}{3} &\mbox{if} \hspace{0.1cm} C(\{2,3\})< R<\Cu, \\
1 & \mbox{if} \hspace{0.1cm} R > \Cu.
\end{cases}
\end{align*}
%we see that the limiting upper value of the game evaluates to 
%\begin{align*}
%L(R)=\begin{cases} \frac{1}{3} &\mbox{if} \hspace{0.1cm} \Cl< R \leq C(\{1,2\})\\
%\frac{1}{2} &\mbox{if}\hspace{0.1cm} C(\{1,2\})< R \leq C(\{1,3\})\\
%\frac{1}{2} &\mbox{if}\hspace{0.1cm} C(\{1,3\}) < R \leq C(\{2,3\})\\
%\frac{2}{3} &\mbox{if}\hspace{0.1cm} C(\{2,3\})< R \leq \overline{C}.
%\end{cases}
%\end{align*}
%and the limiting lower value of the game evaluates to
%\begin{align*}
%\epsilon(R)=\begin{cases} \frac{1}{3} &\mbox{if} \hspace{0.1cm} \Cl \leq R<C(\{1,2\})\\
%\frac{1}{2} &\mbox{if} \hspace{0.1cm} C(\{1,2\})\leq R<C(\{1,3\})\\
%\frac{1}{2} &\mbox{if}\hspace{0.1cm} C(\{1,3\}) \leq R<C(\{2,3\})\\
%\frac{2}{3} &\mbox{if} \hspace{0.1cm}C(\{2,3\})\leq R<\overline{C}.
%\end{cases}
%\end{align*} 
%As can be seen from Fig~\ref{}, except for finitely many rates ($R=\Cl, C(\{1,2\}),C(\{2,3\}),\Cu $) a limiting saddle point value exists for the game with a three-state jammer.
\begin{figure}
\includegraphics[scale=0.46,clip=true, trim=0.21in 2.5in 0in 1.1in]{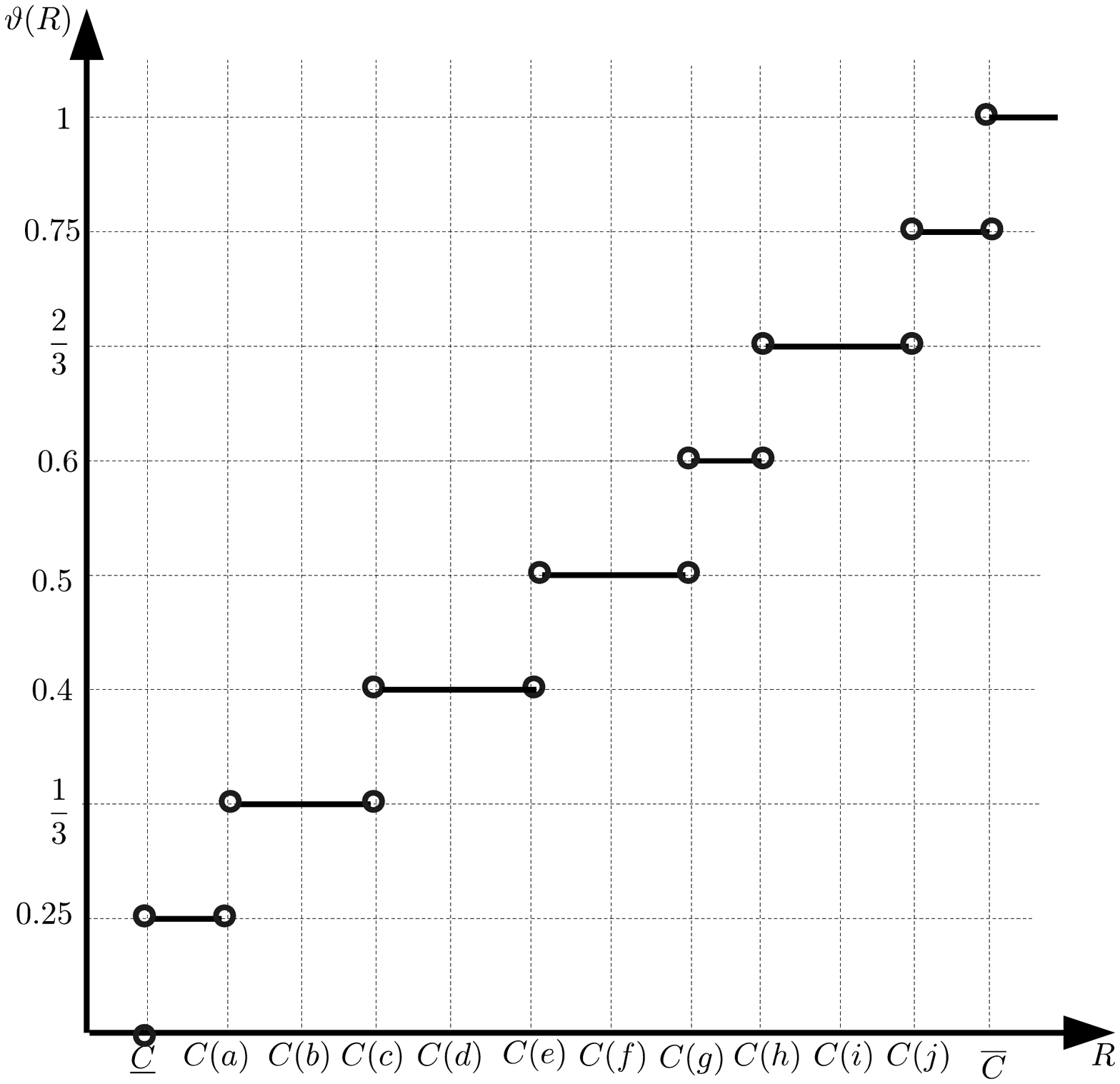}
  \caption{$\vartheta(R)$ for four-state jammer where $C(\Theta')$'s for $\Theta' \subset \Theta=\{1,2,3,4\}$ are assumed to satisfy: $\Cl<C(a)<C(b)<C(c)<C(d)<C(e)<C(f)<C(g)<C(h)<C(i)<C(j)<\Cu=C(\{1\})$ where $a=\{1,2,3\}, b=\{1,2,4\}, c=\{1,3,4\}, d=\{1,2\}, e=\{1,3\}, f=\{1,4\}, g=\{2,3,4\}, h=\{2,3\}, i=\{2,4\},j=\{3,4\}$. The hollow circle indicates an excluded end-point rate.}
  \label{fig:Theta4}
  \end{figure}

\noindent The saddle point value for $|\Theta|=4$ is shown in Fig \ref{fig:Theta4}.

To show our results, we find two categories of bounds: a lower bound on $\lgame$, which is equivalently a finite blocklength \textit{converse} for the mixed channel, and an upper bound on $\ugame$, which is equivalently a finite blocklength \textit{achievability} result for the compound channel. The former is obtained via the linear programming (LP) approach we introduced in~\cite{jose2016linearIT}.  For the latter, in the regimes $R<\Cl$ and $R>\Cu$, we employ a deterministic achievability scheme. In the intermediate region $\Cl<R<\Cu$, we employ a stochastic construction that randomly chooses between codebooks for compound channels formed by subsets of $\Theta$. This achievability scheme is novel in the sense that it was derived using the \textit{converse} as a reference, reinterpreting the latter via von Neumann's minimax theorem, then and attempting to match it, rather than the common approach which attempts to find a matching converse for an achievability scheme. 
\subsection{Related Work}
The existence of a saddle point solution for a game comprising of a communication system and a jammer has been studied multiple times in the LQG setting (\eg, \cite{basar1983gaussian}, \cite{basar1985complete}). For transmitting Gaussian sources over Gaussian channels controlled by a power-constrained jammer with access to the source, the existence of a mixed saddle point solution is shown in \cite{basar1986solutions}. An extension of the above work to general sources over additive noise channel has been considered in \cite{akyol2013optimal}. Closely related to our problem setup is the Arbitrarily Varying Channel (AVC) model (see \cite{lapidoth1998reliable}), where a jammer changes the state of the channel at each instant during the transmission of the encoded message. AVC is used to model packet-dropout adversarial channels in \cite{langbort2012one}, where the control problem over an adversarial channel is then posed as a zero-sum game between the jammer and the controller. For gaussian AVC's with power constraints, the problem of evaluating the asymptotic random coding capacity is shown to have an equivalent zero-sum game formulation between the transmitter and the jammer in \cite{hughes1987gaussian}. To the best of our knowledge ours is the only work that studies the setting we consider, where the channel is the compound channel.
%Evaluating the capacity of this channel model can be then posed as the zero-sum game between the encoder-decoder and the jammer. 

\subsection{Organization}
Section~\ref{sec:infobackground} formulates the zero-sum game between the finite blocklength communication system and the finite state jammer. We also give a background on zero-sum games and 
% explains the relation of the min-max value and max-min value of the game to the finite blocklength channel coding of 
 compound and mixed channels. In Section~\ref{sec:LP}, we present the LP-relaxation based framework to obtain a lower bound on the max-min value of the game. Section~\ref{sec:asymptextreme} presents a new upper bound on the min-max value of the game and analyses the asymptotic tightness of this bound to the lower bound obtained via LP for rates $R<\Cl$ and $R<\Cu$. In Section~\ref{sec:asymtinter}, we consider the intermediate rate region $\Cl<R<\Cu$. We derive a novel finite blocklength upper bound on the min-max value of the game and show that this bound is asymptotically equal to the LP-based lower bound for all but finitely many rates in this range. Finally, we conclude with Section~\ref{sec:conclusion}.
\subsection{Notation}
Throughout this paper, upper case letters $A,B$ represent random variables taking values in spaces represented by calligraphic letters $\Ascr, \Bscr$ respectively and lower case letters $a,b$ represent the specific values these random variables take. Let $\Ibb\{\bullet\}$ represent the indicator function which is one when $\bullet$ is true and is zero otherwise. $\Pscr(\bullet)$ denotes the set of all probability distributions on `$\bullet$'. Let
$\Zscr:=\Sscr \times \Xscr \times \Yscr \times \Sscrhat$ and $z:=(s,x,y,\shat) \in \Zscr$. Let $Q_{X|S}Q_{\Shat|Y}P_{Y|X,\theta}(z,\theta)\equiv Q_{X|S}(x|s)Q_{\Shat|Y}(\shat|y)P_{Y|X,\theta}(y|x,\theta)$. If $P$ represents an optimization problem, $\OPT(P)$ represents its optimal value. If $\Ascr$ represents a set, $\Ascr^c$ represents its compliment. Let L(R)HS represent Left (Right) Hand Side.
Let $x \in \Ascr^n$ be a $n$-length string. Then,
$$P(a|x)=\frac{\sum_{i=1}^n \Ibb\{x_i=a\}}{n}, \quad a \in \Ascr,$$ represents the type of $x$ (see \cite{csiszar2011information}). Note that $P(a|x)\geq 0$ for all $a \in \Ascr$ and $\sum_{a \in \Ascr}P(a|x)=1$, whereby $P \in \Pscr(\Ascr)$.
 %The set of all $n$-length strings of type $P$ in $\Ascr^n$ is represented by $T_P^n$. 
 Let $\Tscr^n$ represent the set of all types in $\Ascr^n$. From type counting lemma \cite{csiszar2011information}, we then have,
\begin{align}
|\Tscr^n|\leq (n+1)^{|\Ascr|}.\label{eq:typecounting}
\end{align}
 
\section{Background and Problem Formulation}\label{sec:infobackground}
%Consider the game between a finite state jammer and a communication system comprising of a team of encoder and decoder. 
%(together called a \textit{code}).
\subsection{Zero-sum games}
A zero-sum game comprises of two players $P_1, P_2$, with strategies denoted $z_1 \in Z_1 $ and $z_2 \in Z_2$, respectively, and a function $c: Z_1 \times Z_2 \rightarrow \Real$ of these strategies. $P_1$ attempts to choose $z_1^* \in Z_1$ so as to minimize $c$ and $P_2$ attempts to choose $z_2^* \in Z_2$ so as to maximize this function, but the value each gets depends not only what the player chooses but also what the other player chooses. 
How would these players then play? von Neumann surmised that $P_1,P_2$ would each choose $z_1^*,z_2^*$, respectively, so as to minimize the worst case damage that the other player could do. \ie, 
\begin{equation}
z_1^* \in \arg \min_{z_1\in Z_1} \overline{c}(z_1), \quad \overline{c}(z_1):= \max_{z_2\in Z_2} c(z_1,z_2), \label{eq:sp1} 
\end{equation}
and
\begin{equation}
z_2^* \in \arg \max_{z_2\in Z_2}\underline{c}(z_2), \quad \underline{c}(z_2):=\min_{z_1\in Z_1} c(z_1,z_2). \label{eq:sp2} 
\end{equation}
%Moreover, $f(x^*,y^*)= \underline{f}(y^*)=\overline{f}(x^*).$ 
%Conditions \eqref{eq:sp1}-\eqref{eq:sp2} characterize the benefit a player would received in the worst case over all strategies of its adversary. In general a pair $(x^*,y^*)$ satisfying \eqref{eq:sp1}-\eqref{eq:sp2} need not satisfy \eqref{eq:sp0}. However, if the for this pair we have $\overline{f}(x^*) = \underline{f}(y^*)$, then they do form a saddle point. 
%Thus we see that at a saddle point, the game reaches an impasse: the `worst case' calculations made both players coincide and neither player can get a better value for $f$ than that which is given by the saddle point, without the other player deviating from the saddle point strategy. 
%%As a result, a saddle point, when it exists, can justifiably be thought of as a solution of the game.

%The worst case values of $c$ guaranteed by these strategies $\overline{c}(z_1^*)$, and $\underline{c}(z_2^*)$, are referred to as the \textit{upper and lower values} of the game and they satisfy $\overline{c}(z_1^*) \geq \underline{c}(z_2^*)$.
Clearly, these calculations may not match up. Specifically, if $P_1$ plays $z_1^*$ it may not be optimal for $P_2$ to play $z_2^*$, or if $P_2$ plays $z_2^*$, $P_1$ may not want to stick to playing $z_1^*.$ Remarkably, von Neumann showed that in certain classes of games~\cite{neumann44games}, these calculations \textit{do} match up exactly and we get $\overline{c}(z_1^*)=\underline{c}(z_2^*)$.
In this case, we get the central solution concept in zero-sum games, proposed by von Neumann, namely, a \textit{saddle point}. $(z_1^*,z_2^*)$ are said to form a saddle point of the game if  
\begin{equation}
c(z_1^*,z_2) \leq c(z_1^*,z_2^*) \leq c(z_1,z_2^*) \quad \forall z_1\in Z_1, z_2\in Z_2.  \label{eq:sp0} 
\end{equation}
The existence of a saddle point is equivalent to the following interchangeability of the `min' and `max' operations on $c$, \ie,
\begin{equation}
\min_{z_1\in Z_1} \max_{z_2\in Z_2}c(z_1,z_2) = \max_{z_2\in Z_2}\min_{z_1\in Z_1}c(z_1,z_2).\label{eq:sp3} 
\end{equation}
In general, one always has that the left hand side of \eqref{eq:sp3} (called the \textit{upper value} of the game, and equal to $\overline{c}(z_1^*)$) is greater than or equal to the right hand side (called the \textit{lower value}, and equal to $\underline{c}(z_2^*)$). 
von Neumann's \textit{minimax theorem}~\cite{neumann44games} showed that equality holds when $Z_1,Z_2$ are finite dimensional probability simplices and $c(z_1,z_2) \equiv z_1\t A z_2$ for a  matrix $A$. 

\subsection{Communication in the presence of a jammer as a game}
We now present the problem of channel coding in presence of a jammer as a zero-sum game. Let $S$ represent the random source message distributed uniformly on $ \Sscr=\{1,\hdots,M\}$ with $M=2^{nR}$ and $R$ is the rate of transmission. 
Consider a family of channels $P_{\Ybb|\Xbb,\theta} \in \Pscr(\Bscr|\Ascr)$ with common finite input and output alphabets $\Ascr$ and $\Bscr$, respectively, parametrized by a state $\theta$ taking values in a finite set $\Theta$. 
Define  $\Xscr:=\Ascr^n$ and $\Yscr:=\Bscr^n$ as the spaces of $n$-length strings of channel inputs and outputs of any channel from this family. An encoder maps $S$ to an $n$-length ($n <\infty$) string $X \in\Xscr$ randomly according to the distribution $Q_{X|S}\in \Pscr(\Xscr|\Sscr)$. $X$ is subsequently sent through a channel in state $\theta \in \Theta$, represented by the stochastic kernel $P_{Y|X,\theta}$ to get the channel output $Y \in \Yscr$.
For each $\theta \in \Theta$, we assume that the channel is discrete and memoryless, \ie
\begin{align}
P_{Y|X,\theta}(y|x,\theta)=\prod_{i=1}^n P_{\Ybb|\Xbb,\theta}(y_i|x_i,\theta),\hspace{0.1cm} \forall x \in \Xscr, y \in \Yscr. \label{eq:DMC}
\end{align}
Subsequently, a decoder maps the channel output $Y$ randomly accordingly to $Q_{\Shat|Y} \in \Pscr(\Sscr|\Yscr)$ to get $\Shat \in \Sscr$. Here, $Q_{X|S}$ represents a \textit{stochastic encoder}, $Q_{\Shat|Y}$ represents a \textit{stochastic decoder} and together, they constitute a \textit{stochastic code} in the sense of Csiszar and Korner~\cite{csiszar2011information} or \textit{behavioral strategies} in the language of game theory~\cite{maschler2013game}. Note that this is distinct from a \textit{random code}. A deterministic code is a pair of functions $f:\Sscr \rightarrow \Xscr, g:\Yscr \rightarrow \Sscr$ and a random code is a randomly chosen pair of functions $f,g$. A jammer controls the state of the channel and chooses a state $\theta \in \Theta$ randomly according to some distribution $q \in \Pscr(\Theta)$ independent of $S,X$. State $\theta$ once chosen is held fixed through $n$-uses of the channel and the channel state chosen is not known to both the encoder and decoder.
 
This setting can be formulated as a zero-sum game  by taking $P_1$ as a \textit{team} comprising of the stochastic encoder and decoder $(Q_{X|S},Q_{\Shat|Y})$ and its strategy space as $Z_1=\Pscr(\Xscr|\Sscr) \times \Pscr(\Sscr|\Yscr)$, the set of all stochastic codes. 
 The jammer is player $P_2$ with its strategy space as $\Pscr(\Theta)$ and the cost function $c$ is the 
 average probability of error over all messages, $\Ebb[\Ibb\{S \neq \Shat\}]$ evaluated under the probability distribution induced by the strategies of the encoder-decoder team and the jammer.  
Since for each value of $\theta$, the random variables $S, X, Y,\Shat$ form a Markov chain in that order, we have
\begin{align}
\Ebb[\Ibb&\{S \neq \Shat\}|S=s,\theta] \non \\
&=\sum_{x,y,\shat}Q_{X|S}(x|s)Q_{\Shat|Y}(\shat|y)P_{Y|X,\theta}(y|x,\theta)\Ibb\{s\neq \shat\}, \label{eq:errorexp} 
\end{align}
and hence 
\begin{align*}
\Ebb[\I{S \neq \Shat}|S=s] &= \sum_{\theta\in \Theta} q(\theta)\Ebb[\Ibb\{S \neq \Shat\}|S=s,\theta],\\ 
\aur \quad \Ebb[\I{S\neq \Shat}] &= \frac{1}{M}\sum_{s \in \Sscr} \Ebb[\I{S \neq \Shat}|S=s].
\end{align*}

The central quest in this paper is to investigate  whether a minimax theorem can be claimed for this game. However, the answer to this question is in the negative in general when $n$ is finite since  while $\Ebb[\Ibb\{S \neq \Shat\}]$ is concave (in fact, linear) in $q$, it is nonconvex in $(Q_{X|S},Q_{\Shat|Y})$, the nonconvexity arising due to the presence of bilinear products $Q_{X|S}(x|s)Q_{\Shat|Y}(\shat|y)$ (see \eg,~\cite{kulkarni2014optimizer}) in~\eqref{eq:errorexp}. 
Consequently, a saddle point value may not exist for this game.
In this context, we explore an `approximate' minmax theorem for the zero-sum game, the approximation becoming increasingly accurate as blocklength $n$ increases. Towards this, we derive new upper bounds on $\ugame$ and lower bounds on $\lgame$ and show that a near-saddle point holds for the zero-sum game, \ie, for each $n \in \Nbb$, $0\leq \ugame-\lgame\leq \epsilon_n$, $\epsilon_n \in [0,1]$, where for all but finitely many values of the rate $R$, we have $\limn \epsilon_n =0$. Consequently, for such rates we get 
$$\limn \ugame=\limn\lgame=:\vartheta(R),$$  
where $\vartheta(R)$ is now the limiting  saddle-point value of the zero-sum game.
\subsection{Background on mixed and compound channels}
Our problem is intimately related to \textit{compound} and \textit{mixed} channels. We recall some background about them here and relate them to our problem~\cite{lapidoth1998reliable}. 
A compound channel is defined by a family of channels $\{P_{\Ybb|\Xbb,\theta}\}_{\theta \in \Theta}$ whose state $\theta$ is held fixed throughout the transmission. The average probability of error incurred by a code $Q_{X|S},Q_{\Shat|Y}$ over this compound channel is defined as
$$\epsilon_{\com}(Q_{X|S},Q_{\Shat|Y};n):=\max_{\theta \in \Theta} \sum_s \frac{1}{M}\Ebb[\Ibb\{S \neq \Shat\}|S=s,\theta].$$
For any $q\in \Pscr(\Theta)$, a mixed channel of blocklength $n$ is defined as $$P_{Y|X}^{(q)}(y|x)=\sum_{\theta \in \Theta}q(\theta)P_{Y|X,\theta}(y|x,\theta),$$ 
where $P_{Y|X,\theta}(y|x,\theta)$ is as defined in \eqref{eq:DMC}. The average probability of error incurred by a code in a mixed channel is defined in the usual sense as 
$$\epsilon_{\mix}(q,Q_{X|S},Q_{\Shat|Y};n):= \Ebb[\Ibb\{S\neq \Shat\}]. $$

From the linearity of $\Ebb[\Ibb\{S \neq \Shat\}]$ in $q$, it is easy to see that problem $\uP$ is equivalent to the finite blocklength channel coding of the compound channel under the average probability of error criterion (where the average is taken over all messsages), employing stochastic codes. Moreover, for each strategy $q \in \Pscr(\Theta)$ of the jammer, the inner minimization in problem $\lP$ corresponds to the finite blocklength channel coding of the mixed channel under the average probability of error criterion.
It is then evident that the bound $\epsilon_n$ mentioned at the end of the previous section can be established using a combination of a finite blocklength achievability for a compound channel and a finite blocklength converse for a mixed channel. Our main contribution is in showing that these asymptotically the same.

%We now quickly recall some background for these channels.
% Consider the channel coding of the compound channel $\{P_{Y|X,\theta}\}_{\theta \in \Theta}$ and the mixed channel $P_{Y|X}^{(q)}$, $q \in \Pscr(\Theta)$ under the average probability of error criterion employing stochastic codes.
% Let $\epsilon_{\com}(Q_{X|S},Q_{\Shat|Y};n):=\max_{\theta \in \Theta} \sum_s \frac{1}{M}\Ebb[\Ibb\{S \neq \Shat\}|S=s,\theta]$ and $\epsilon_{\mix}(q,Q_{X|S},Q_{\Shat|Y};n)=: \Ebb[\Ibb\{S\neq \Shat\}]$ represent the average probability of error incurred in the compound and mixed channels respectively, from employing a stochastic code $(Q_{X|S},Q_{\Shat|Y})$.
A rate $R$ is said to be achievable for a compound channel (or mixed channel) if there exists a sequence of stochastic codes $(Q_{X|S},Q_{\Shat|Y})$ with $M=2^{nR}$ such that $\epsilon_{\com}(Q_{X|S},Q_{\Shat|Y};n)$ (or $\epsilon_{\mix}(q,Q_{X|S},Q_{\Shat|Y};n)$)  goes to $0$ as $n \rightarrow \infty$. The supremum over all such achievable rates then gives
the \textit{capacity} of a compound channel (or mixed channel) under stochastic codes.
%$\rightarrow$ change this to consider coding under stochastic codes and average prob of error.}
% is defined as the maximum value of $R$ achievable over all determinstic codes $(f,g)$ such that $\epsilon_{\com}(f,g;n)$ (or $\epsilon_{\mix}(q,f,g;n)) $ goes to $0$ as $n \rightarrow \infty$.
 Interestingly, the capacity (under both deterministic and stochastic codes) of a compound channel is (see \cite{blackwell1960capacities}) 
 $$ \Cl= \max_{P_{\Xbb}}\min_{\theta \in \Theta}I_{P_{\Xbb}}(\Xbb;\Ybb|\theta),$$  and the capacity of the mixed channel is (see~\cite{lapidoth1998reliable}) $$\max_{P_{\Xbb}}\min_{\theta: q(\theta)>0}I_{P_{\Xbb}}(\Xbb;\Ybb|\theta),$$  
where the maximization is over $P_{\Xbb}\in \Pscr(\Ascr)$ and $I_{P_{\Xbb}}(\Xbb;\Ybb|\theta):=$  
$$\sum_{a \in \Ascr, b \in \Bscr} P_{\Xbb}(a)P_{\Ybb|\Xbb,\theta}(b|a,\theta) \log \frac{ P_{\Ybb|\Xbb,\theta}(b|a,\theta)}{\sum_{a\in \Ascr}P_{\Xbb}(a)P_{\Ybb|\Xbb,\theta}(b|a,\theta)}, $$ represents the mutual information between random variables $\Xbb \in \Ascr$ and $\Ybb \in \Bscr$. Notice that the mixed channel capacity depends only on the support of $q$, and when the support is $\Theta$, it equals $\Cl$, the capacity of the compound channel. We assume throughout this paper that $\Cl>0.$

Under the average probability of error criterion, a strong converse does not hold for mixed and compound channels (see \cite{ahlswede1967certain,ahlswede1968weak}). 
% This in turn, has led to studies on evaluation of the $\epsilon$-capacity of compound \cite{ahlswede2016transmitting} (and mixed channels), which is the supremum of all achievable rates such that $\limn \epsilon_{\com}(Q_{X|S},Q_{\Shat|Y};n) \leq \epsilon$ (and $\limn \epsilon_{\mix}(q,Q_{X|S},Q_{\Shat|Y};n) \leq \epsilon$), $\epsilon \in [0,1)$.  As such, for the mixed channel $P_{Y|X}^{(q)}$, the $\epsilon$-capacity under average error probability has been obtained in \cite[Thm 1 and Lem 1(a)]{yagi2014single} as, $C_{\epsilon}^{(q)}=$
%\begin{align}
%\max_{P_{\Xbb} \in \Pscr(\Ascr)} \sup \biggl \lbrace R \mid \sum_{\theta}q(\theta)\Ibb\{I_{P_{\Xbb}}(\Xbb;\Ybb|\theta)\leq R\}\leq \epsilon \biggr \rbrace. \label{eq:epscapacity}
%\end{align}
However, if the average probability of error criterion is replaced with maximum probability of error (over all messages), \ie, $\max_{s\in \Sscr,\theta\in \Theta} \Ebb[\Ibb\{S \neq \Shat\}|S=s,\theta]$, a strong converse holds for the compound channel (\cite{lapidoth1998reliable}) but not the mixed channel (for the latter the criterion is $\max_{s\in \Sscr} \Ebb[\Ibb\{S \neq \Shat\}|S=s]$). We will see that this subtle difference plays an important role in our result.
%\textbf{Include Yagi and Nomura $\epsilon$-capacity results here. Mention that there is no strong converse under average prob of error for mixed and compound. Similarly strong converse for max prob of error for compound but not for mixed.} 
%Consequently, when $R<\Cl$, 
%$\min_{Q_{X|S}, Q_{\Shat|Y}}\epsilon_{\com}(Q_{X|S},Q_{\Shat|Y};n)$ goes to $0$ as $n \rightarrow \infty$. Our bounds and asymptotic analysis imply this fact.

%While capacity $\Cl$ is defined as the supremum over all achievable rates such that the average probability of error vanishes as $n \rightarrow \infty$, similar notion of capacity exists when the error probability is bound to not exceed $\epsilon \in (0,1)$ as $n \rightarrow \infty$. This results in what is known as the $\epsilon$- capacity of the channel. For a mixed channel $P_{Y|X}^{(q)}$, the $\epsilon$-capacity is obtained in \cite{yagi2014single} as
%\begin{align}
%C(\epsilon,q)=\sup_{P_{\Xbb}} \sup \biggl \lbrace R \mid \sum_{\theta}q(\theta)\Ibb\{I_{P_{\Xbb}}(\Xbb;\Ybb|\theta)\leq R\}\leq \epsilon \biggr \rbrace, \label{eq:epscapacity}
%\end{align}

%\section{Preliminaries}
%\subsection{Types and typical sequences}

\section{Finite Blocklength Lower Bounds}\label{sec:LP}
In this section, we obtain a lower bound on $\lgame$ via a linear programming relaxation of the minimization part of problem $\uP$.  For each $q \in \Pscr(\Theta)$, $\Ebb[\Ibb\{S \neq \Shat\}]$ is nonconvex in $(Q_{X|S},Q_{\Shat|Y})$, the nonconvexity resulting from the presence of bilinear products. Consequently, we relax this nonconvexity by a linear programming (LP) relaxation which results in a new min-max problem over \textit{relaxed codes}.

Towards this, we resort to a lift-and-project-like idea as illustrated in \cite{jose2016linearIT}. For sake of completeness, we quickly outline the technique here. In this approach, we introduce new variables,  $W(s,x,y,\shat)$ to replace the bilinear product terms $Q_{X|S}(x|s)Q_{\Shat|Y}(\shat|y)$ for each $s\in \Sscr, x \in \Xscr, y \in \Yscr, \shat \in \Sscrhat$, thereby lifting the nonconvex problem to a higher dimensional space of variables, that now includes $(Q_{X|S},Q_{\Shat|Y},W)$. New valid inequalities in terms of $W$ are obtained in this space. Towards this, for each $s \in \Sscr$, we multiply both sides of the equation $\sum_x Q_{X|S}(x|s)=1$ with $Q_{\Shat|Y}(\shat|y)$ for all $\shat,y$ and for each $y \in \Yscr$, multiply both sides of $\sum_{\shat}Q_{\Shat|Y}(\shat|y)=1$ with $Q_{X|S}(x|s)$ for all $x,s$. Replace the bilinear products in the resulting set of equations with the new variable $W(s,x,y,\shat)$ and add these to the original constraints. 
Thus, a \textit{lower bound} on $\ugame$ is given by the optimal value of the following problem,
$$\problemsmalla{$\LP(n;R)$\hspace{0.1cm}}
	{Q_{X|S},Q_{\Shat|Y},W}
	{\displaystyle \max_q \ \ \err(q,W)}
				 {\begin{array}{r@{\ }c@{\ }l}
				% Q(z)&=&P_S(s)Q_{X|S}(x|s)P_{Y|X,\theta}(y|x,\theta)q(\theta)Q_{\Shat|Y}(\shat|y)\\
				%\sum_{z} P_S(s)W(z)\Ibb\{\shat\neq s\}P_{Y|X,\theta}(y|x,\theta)&\leq & \Delta \forall \theta \hspace{0.1cm} :\mu(\theta)\\
Q_{X|S} \in \Pscr(\Xscr|\Sscr), Q_{\Shat|Y} &\in &\Pscr(\Sscrhat|\Yscr), q \in \Pscr(\Theta) \\
				 				 \sum_x W(z)-Q_{\Shat|Y}(\shat|y)&=& 0\hspace{0.1cm} \hspace{0.25cm} \forall s,\shat,y,\\
				 				 \sum_{\shat}W(z)-Q_{X|S}(x|s)&= &0 \hspace{0.1cm} \hspace{0.2cm} \forall s,x,y,
	\end{array}}$$ 
where $$\err(q,W)=\sum_{z,\theta}\frac{1}{M}W(z)\Ibb\{s \neq \shat\}q(\theta)P_{Y|X,\theta}(y|x,\theta).$$  A collection	$(Q_{X|S},Q_{\Shat|Y},W)$ satisfying the constraints above is called a \textit{relaxed code}.  Clearly $\ugame \geq \OPT(\LP(n;R)).$
% and
%\begin{align*}
%&L:=\biggl \lbrace (Q_{X|S},Q_{\Shat|Y},W) \bigg | \sum_x W(z)\equiv Q_{\Shat|Y}(\shat|y), \\& \sum_{\shat}W(z)\equiv Q_{X|S}(x|s), \sum_x Q_{X|S}(x|s)\equiv 1, \\&\sum_{\shat} Q_{\Shat|Y}(\shat|y)\equiv 1,  Q_{X|S},Q_{\Shat|Y},W\geq 0 \biggr \rbrace.
%\end{align*}
Note that $\err(q,W)$ is convex in $(Q_{X|S},Q_{\Shat|Y},W)$ and linear in $q$. Moreover, the set of relaxed codes is convex and compact, and so is the set of $\Pscr(\Theta)$. Consequently, a minmax theorem holds, and the order of the minimization and maximization can be interchanged in $\LP(n;R)$. Thus,
\begin{align}
\OPT(\LP(n;R))=\max_{q} \min_{(Q_{X|S},Q_{\Shat|Y},W)} \err(q,W), \label{eq:interchange} 
\end{align}
where the minimization is over relaxed codes and the maximization is over $q \in \Pscr(\Theta).$ Now, for each $q \in \Pscr(\Theta),$ the inner  minimization  over relaxed codes in \eqref{eq:interchange} is a lower bound on the inner minimization $\min_{Q_{X|S},Q_{\Shat|Y}}\Ebb[\I{S\neq \Shat}]$ in $\lP$; it is precisely what one would obtain if one applied the above relaxation directly to this minimization. Consequently, $\OPT(\LP(n;R))$ is a lower bound \textit{also} on $\lgame$, whereby,
\begin{align}
\ugame&\geq \lgame
\geq \OPT(\LP(n;R)). \label{eq:lb} 
\end{align}
This substantiates our motivation -- if $\LP(n;R)$ is a tight relaxation of $\uP$, then the upper and lower values of the game ought to approach each other.

To obtain a lower bound from $\LP(n;R)$, we replace the inner minimization over relaxed codes in \eqref{eq:interchange} by its dual. 
The dual of the inner minimization in \eqref{eq:interchange} can be written as,
$$
\maxproblemsmall{$\DP(q,n;R)$}
	{\gamma_a^{(q)},\gamma_b^{(q)},\lambda_{\ssf}^{(q)},\lambda_{\ssf}^{(q)}}
	{\displaystyle \sum_{s}\gamma_a^{(q)}(s)+\sum_{y}\gamma_b^{(q)}(y)}
				 {\hspace{-.8cm}\begin{array}{r@{\ }c@{\ }l}
				 				 \gamma_a^{(q)}(s)- \sum_y \lambda_{\csf}^{(q)}(s,x,y) &\leq& 0 \quad \hspace{0.3cm} \hspace{0.3cm}\quad\forall x,s \quad (D1)\\
\gamma_b^{(q)}(y)- \sum_s \lambda_{\ssf}^{(q)}(s,\shat,y)&\leq& 0  \quad \hspace{0.3cm}\hspace{0.3cm} \quad\forall \shat,y\quad (D2)\\
\lambda_{\ssf}^{(q)}(s,\shat,y)+\lambda_{\csf}^{(q)}(s,x,y)& \leq & \Lambda^{(q)}(z) \quad \quad \forall z \quad (D3)\\
%\sum_{\theta}\mu(\theta)\leq 1, \quad
%\mu(\theta) &\geq& 0 \qquad  \hspace{0.3cm}\quad\forall \theta,
 	\end{array}}\vspace{-0.1cm}
	$$ where $$\Lambda^{(q)}(z)=\frac{1}{M}\Ibb\{s \neq \shat\}\sum_{\theta}P_{Y|X,\theta}(y|x,\theta)q(\theta),$$
	and $\lambda_{\ssf}^{(q)}:\Sscr \times \Sscrhat \times \Yscr \rightarrow \Real$, $\lambda_{\csf}^{(q)}: \Sscr \times \Xscr \times\Yscr \rightarrow \Real$, $\gamma_a^{(q)}: \Sscr \rightarrow \Real$ and $\gamma_b^{(q)}: \Yscr \rightarrow \Real$ are Lagrange multipliers (see~\cite{jose2016linearIT} for details). 
	From duality of linear programming, it then follows that $\OPT(\LP(n;R))=\max_{q \in \Pscr(\Theta)} \OPT(\DP(q,n;R))$. Towards evaluating $\OPT(\DP(q,n;R))$, 
note that it is optimal to take $\gamma_a^{(q)}(s)=\min_x \sum_y \lambda_{\csf}^{(q)}(s,x,y)$ and $\gamma_b^{(q)}(y)=\min_{\shat}\sum_s \lambda_{\ssf}^{(q)}(s,\shat,y)$. Consequently, we get $\OPT(\DP(q,n;R))=$
\begin{align*}
&\max_{\lambda_{\csf}^{(q)},\lambda_{\ssf}^{(q)}}\biggl \lbrace \sum_s \min_x \sum_y \lambda_{\csf}^{(q)}(s,x,y)+\sum_y \min_{\shat}\sum_s \lambda_{\ssf}^{(q)}(s,\shat,y)\biggr \rbrace \non\\
&\mbox{s.t} \qquad \lambda_{\ssf}^{(q)}(s,\shat,y)+\lambda_{\csf}^{(q)}(s,x,y) \leq  \Lambda^{(q)}(z) \quad \forall z.
\end{align*}
The following lemma then outlines our framework for deriving lower bound on $\lgame$.
\begin{lemma}\label{lem:framework}
Any choice of functions $\lambda_{\ssf}^{(q)}:\Sscr \times \Sscrhat \times \Yscr \rightarrow \Real$, $\lambda_{\csf}^{(q)}:\Xscr \times \Sscr \times \Yscr \rightarrow \Real$ satisfying constraint (D3) yields the following lower bound on $\lgame$,
\begin{align*}
& \lgame
\geq \max_{q}\OPT(\DP(q,n;R)) \non\\& \geq \max_{q} \biggl \lbrace \hspace{-0.1cm} \sum_s \min_x \sum_y \lambda_{\csf}^{(q)}(s,x,y) +\sum_y \min_{\shat}\sum_s \lambda_{\ssf}^{(q)}(s,\shat,y)\biggr \rbrace. 
%\label{eq:converseframework}
\end{align*}
\end{lemma}
%The equality in $(a)$ follows from the duality of linear programming.

%Thus, to obtain a lower bound on $\lgame$, it is enough to construct functions $\lambda_{\ssf}^{(q)}:\Sscr \times \Sscrhat \times \Yscr \rightarrow \Real$, $\lambda_{\csf}^{(q)}:\Xscr \times \Sscr \times \Yscr \rightarrow \Real$ satisfying constraint (D3) of $\DP(q)$ for each $q \in \Pscr(\Theta)$.
 In fact, a particular choice of functions $\lambda_{\ssf}^{(q)}$, $\lambda_{\csf}^{(q)}$  satisfying constraint (D3) of $\DP(q,n;R)$ for each $q \in \Pscr(\Theta)$ gives the following lower bound on $\lgame$.
\begin{theorem}
The following lower bound on $\lgame$ holds,
\begin{align}
&\lgame\geq \max_{q}\OPT(\DP(q,n;R))\non\\
&\geq \max_{q} \sup_{P_{\Ybar|\theta},\gamma>0} \biggl \lbrace \min_x \sum_y  \sum_{\theta}q(\theta) \min\biggl \lbrace P_{Y|X,\theta}(y|x,\theta), \non \\&\qquad P_{\Ybar|\theta}(y|\theta)M\exp(-\gamma)\biggr \rbrace  -\exp(-\gamma)\biggr \rbrace, \label{eq:conversegen}
\end{align}where $P_{\Ybar|\theta}$ is an arbitrary probability distribution in the space of $\Yscr$ for $\theta \in \Theta$.
\end{theorem}
\begin{proof}
For each $q \in \Pscr(\Theta)$, consider the following choice of functions $\lambda_{\ssf}^{(q)}, \lambda_{\csf}^{(q)}$,
\begin{align*}
\lambda_{\csf}^{(q)}(s,x,y)&=\sum_{\theta}\frac{q(\theta)}{M}\min\biggl \lbrace P_{Y|X,\theta}(y|x,\theta), M\frac{P_{\Ybar|\theta}(y|\theta)}{\exp(\gamma)}\biggr \rbrace,\\
\lambda_{\ssf}^{(q)}(s,\shat,y)&=-\frac{1}{M}\sum_{\theta}q(\theta)P_{\Ybar|\theta}(y|\theta)M\exp(-\gamma) \Ibb\{s=\shat\}.
\end{align*}
The feasibility of these functions with respect to (D3) of $\DP(q,n;R)$ can be verified as in the proof of Theorem~2.1 \cite{jose2018improvedSW} and we skip the proof here.
%To verify the feasibility of these functions with respect to (D3) of $\DP(q)$, consider the following two cases.
%When $s\neq \shat$, $\lambda_{\ssf}^{(q)}(s,\shat,y)=0$ and the LHS of (D3) evaluates to
%\begin{align*}
%\lambda_{\csf}^{(q)}(s,x,y) \leq  \sum_{\theta}\frac{q(\theta)}{M}P_{Y|X,\theta}(y|x,\theta)
%\end{align*}which is the RHS of (D3).
%For the case when $s=\shat$, the RHS of (D3) becomes zero and the LHS evaluates to
%\begin{align*}
%\frac{1}{M}\sum_{\theta}q(\theta) \biggl [\min\biggl \lbrace P_{Y|X,\theta}(y|x,\theta), z_{y|\theta}\biggr \rbrace -z_{y|\theta}\biggr]  \leq 0,
%\end{align*} with $z_{y|\theta}=P_{\Ybar|\Theta}(y|\theta)M\exp(-\gamma)$, thereby satisfying the constraint (D3). Thus, the considered choice of functions is feasible for $\DP(q)$, $q \in \Pscr(\Theta)$.
Taking supremum over $P_{\Ybar|\theta}(y|\theta)$ and $\gamma > 0$ of the resulting dual cost and applying Lemma~\ref{lem:framework} gives the required lower bound.
\end{proof}
 Lower bounding $\min\biggl \lbrace P_{Y|X,\theta}(y|x,\theta), P_{\Ybar|\Theta}(y|\theta)M\exp(-\gamma)\biggr \rbrace$ in \eqref{eq:conversegen} with $P_{Y|X,\theta}(y|x,\theta)\Ibb \biggl \lbrace i_{X;\Ybar|\theta}(x;Y|\theta)\leq \log M -\gamma \biggl \rbrace,$
%One can further lower bound \eqref{eq:conversegen} by bounding \begin{align*}
%&\min\biggl \lbrace P_{Y|X,\theta}(y|x,\theta), P_{\Ybar|\Theta}(y|\theta)M\exp(-\gamma)\biggr \rbrace \non \\&\qquad \geq P_{Y|X,\theta}(y|x,\theta)\Ibb \biggl \lbrace i_{X;\Ybar|\theta}(x;Y|\theta)\leq \log M -\gamma \biggl \rbrace,
%\end{align*}
where $$i_{X;\Ybar|\theta}(x;y|\theta)=\log \frac{P_{Y|X,\theta}(y|x,\theta)}{P_{\Ybar|\theta}(y|\theta)},$$  then yields the following bound,
\begin{align}
&\max_{q}\sup_{P_{\Ybar|\theta},\gamma>0}\min_x \biggl \lbrace  \sum_{\theta}q(\theta)\biggl[ \Pbb_{\theta}[ i_{X;\Ybar|\theta}(x;Y|\theta)\leq \log M -\gamma] \non \\& -\exp(-\gamma) \biggr \rbrace. \label{eq:converse1}
\end{align}
%Together with the following lemma from \cite{yagi2014single}, 
\eqref{eq:converse1} in turn results in the following lower bound on $\lgame$.
% in Corollary~\ref{cor:typebasedlb}.
\begin{corollary}\label{cor:typebasedlb}
The following lower bound on $\lgame$ holds,
\begin{align}
&\lgame\geq \max_{q}\OPT(\DP(q,n;R))\non\\
&\geq \max_{q \in \Pscr(\Theta)} \sum_{\theta}\hspace{-0.1cm}q(\theta)\Ibb\biggl\lbrace \hspace{-0.1cm}I_{\Pbar_n(q)}(\Xbb;\Ybb|\theta)\hspace{-0.1cm} \leq R -2\xi-\frac{\log |\Tscr^n|}{n}\hspace{-0.05cm}\biggr\rbrace \non\\&\qquad \qquad - \frac{A(\xi)}{n}-\exp(-n \xi), \label{eq:typelb}
\end{align}where $\xi>0$ is an arbitrary constant,
\begin{align} P_{\Ybar|\theta}(y|\theta)\equiv \frac{1}{|\Tscr^n|}\sum_{P_n\in \Tscr^n} (P_nP_{\Ybb|\Xbb,\theta})^{\times n}(y),\label{eq:pybar}\end{align} and  $\Pbar_n(q)$ is the type of $x \in \Ascr^n$ that minimizes
\begin{align}
\sum_{\theta}q(\theta) \sum_y P_{Y|X,\theta}(y|x,\theta) \Ibb\biggl \lbrace \log \frac{P_{Y|X,\theta}(y|x,\theta)}{P_{\Ybar|\theta}(y|\theta)} \leq nR-n\xi\biggr \rbrace. \label{eq:Pnq}
\end{align} 
\end{corollary}
\begin{proof}
The proof is included in Appendix~\ref{app:3}.
\end{proof}
%In particular, taking $z_{y|\theta} =P_{\Ybar|\theta}(y|\theta)M\exp(-\gamma)$, where $P_{\Ybar|\theta}$ is an arbitrary probability distribution in the space of $\Yscr$ and $\gamma>0$ is an arbitrary constant, we recover the following lower bound for mixed channels obtained in \cite{jose2017linearaveraged},
%\begin{align}
%&\lgame \geq \max_{q(\theta)}\OPT(\DP(q))\non\\
%&\geq \max_{q(\theta)}\sup_{P_{\Ybar|\theta},\gamma>0}\min_x \biggl \lbrace  \sum_{\theta}q(\theta)\sum_y \biggl[P_{Y|X,\theta}(y|x,\theta)\times \non \\& \Ibb \biggl \lbrace i_{X;\Ybar|\theta}(x;y|\theta)\leq \log M -\gamma \biggl \rbrace+P_{\Ybar|\theta}(y|\theta)M\exp(-\gamma)\times \non \\&\Ibb \biggl \lbrace i_{X;\Ybar|\theta}(x;y|\theta)> \log M -\gamma \biggl \rbrace \biggr] -\exp(-\gamma) \biggr \rbrace, \label{eq:converse1}
%\end{align}
%where $$i_{X;\Ybar|\theta}(x;y|\theta)=\log \frac{P_{Y|X,\theta}(y|x,\theta)}{P_{\Ybar|\theta}(y|\theta)}.$$
%\begin{remarkc}\label{rem:mixed}
%In the information theory setting, $\sum_{\theta}q(\theta)P_{Y|X,\theta}(y|x,\theta)$ represents an 
%\textit{averaged or mixed} channel, first introduced by Ahlswede in \cite{ahlswede1968weak}. Consequently, for each $q \in \Pscr(\Theta)$, the max-min problem $\lgame$ is equivalent to minimizing the average probability of error incurred over the averaged channel averaged according to the distribution $q(\theta)$. Under maximum probability of error criterion, a lower bound for mixed channels has been derived employing the LP relaxation in \cite{jose2017linearaveraged}.
%\end{remarkc}\\
Having obtained lower bounds on the lower value of the game, we now move on to analysing the asymptotic equality of $\ugame$ and $\lgame$ in the limit as $n \rightarrow \infty$. Towards this, we divide the rate region into two different cases -- $(a)$ $R>\Cu$ or $R<\Cl$, $(b)$ $\Cl<R<\Cu$ -- and we obtain new upper bounds on the $\ugame$ in these regions. Together with the LP-based lower bound, we then show in Sections~\ref{sec:asymptextreme} and \ref{sec:asymtinter} that the min-max and max-min values of the game indeed approach each other in the limit as $n \rightarrow \infty$ in these regions.
\section{Asymptotic Equality of $\ugame$ and $\lgame$ when $R<\Cl$ and $R>\Cu$}\label{sec:asymptextreme}
In this section, we analyze the asymptotic equality of $\ugame$ and $\lgame$ in the limit as $n \rightarrow \infty$ when $R<\Cl$ and $R>\Cu$. Since a lower bound on $\lgame$ follows from the LP relaxation-based bound in \eqref{eq:typelb}, we first derive a new finite blocklength upper bound (or achievability bound) on the min-max value $\ugame$.

%\subsection{Finite blocklength Upper Bounds}
%This section presents a new  upper bound on $\ugame$. 
Towards this, note that the objective value of $\uP$ corresponding to any choice of distributions $Q_{X|S},Q_{\Shat|Y}$ yields an upper bound on $\ugame$. In particular, we construct a deterministic code $(f,g)$ so as to obtain the following upper bound on $\ugame$.
 %\ie $Q_{X|S}(x|s)=\Ibb\{x=f(s)\}$, $Q_{\Shat|Y}(\shat|y)=\Ibb\{\shat=g(y)\}$, where 

%The decoder $g$ then states that a message $i$ is sent if the channel ouput corresponding to $u_i$ lies in $\Dscr_i$, $i =1,\hdots,M.$ 
% A specific construction of this codebook and decoding sets yield the following upper bound on $\ugame$.
\begin{theorem}\label{thm:achievability}
Let $P_X \in \Pscr(\Xscr)$ be any input distribution. Then, for any stochastic code, the following upper bound holds, 
\begin{align}
\ugame&\leq \max_{\theta \in \Theta, s\in \Sscr} \Ebb[\Ibb\{S \neq \Shat\}|S=s,\theta] \non \\ 
&\leq \max_{\theta \in \Theta}\Pbb_{\theta}\biggl[ \log \frac{P_{Y|X,\theta}(Y|X,\theta)}{P_{Y|\theta}(Y|\theta)}\leq \alpha+\delta \biggr]\non\\&\qquad +\exp(-\delta)+\frac{M|\Theta|^2}{\exp(\alpha)},\label{eq:improvedfeinstein}
\end{align}
where $\Ebb[\Ibb\{S \neq \Shat\}|S=s,\theta]$ is the probability of error in transmitting message $s$ under this code when the channel state is $\theta$, 
$\Pbb_{\theta}$ is the probability with respect to $P_XP_{Y|X,\theta}$, $\alpha,\delta>0$ are arbitrary constants and $P_{Y|\theta}(y|\theta)=\sum_{x \in \Xscr} P_X(x) P_{Y|X,\theta}(y|x,\theta)$.
\end{theorem}
\begin{proof}
Proof is included in Appendix~\ref{app:2}.
\end{proof}
In particular, fixing 
%$P_X(x)=\prod_{i=1}^n P_{\Xbb}(x_i) \quad \forall x \in \Ascr^n$, where $P_{\Xbb}\in \Pscr(\Ascr)$ and 
$\alpha=\log M+ \frac{n\Delta}{2}$ , $\delta=\frac{n\Delta}{2}$ with $\Delta>0$ an arbitrary positive constant, \eqref{eq:improvedfeinstein} results in the following upper bound,
\begin{align}
\ugame&\leq \max_{\theta \in \Theta}\Pbb_{\theta}\biggl[ \log \frac{P_{Y|X,\theta}(Y|X,\theta)}{P_{Y|\theta}(Y|\theta)}\leq \log M+n\Delta \biggr]\non\\&\qquad +\exp(-n\Delta/2)[|\Theta|^2+1].\label{eq:improvedfeinstein1}
\end{align}

%Note that optimal value of CC can be equivalently expressed as $\OPT(\ugame)=$
%$$\min_{\substack{Q_{X|S} \in \Pscr(\Xscr|\Sscr),\\ Q_{\Shat|Y} \in \Pscr(\Sscrhat|\Yscr)}} \max_{\theta \in \Theta}\sum_z \frac{1}{M}Q_{X|S}Q_{\Shat|Y}P_{Y|X,\theta}(z,\theta)\Ibb\{s\neq \shat\},$$ which follows since $\Ebb[\Ibb\{S \neq \Shat\}]$ is linear in $q(\theta)$. It is then easy to see that problem $\ugame$ is equivalent to the problem of minimizing the average probability of error over randomized codes incurred in a compound channel (see \cite{el2011network}), characterized by the collection $\{P_{Y|X,\theta}, \theta \in \Theta\}$. 
Recall from Section~\ref{sec:infobackground} that evaluating $\ugame$ is equivalent to the problem of minimizing the average probability of error  over all stochastic codes  in a compound channel.
%However, in the setting of $\uP$, neither of the mappings $Q_{X|S},Q_{\Shat|Y}$ knows of the actual channel conditional distribution employed for the transmission of $X$. 
%Consequently, finding an upper bound on $\ugame$ via deterministic codes translates to deriving an upper bound for the finite blocklength coding of compound channels when no channel state information is available to both encoder and decoder. 
Earliest upper bound known for such compound channels has been obtained by Blackwell \etal \cite{blackwell1959capacity}. 
%via codebook construction for a uniform averaged channel and then generalizing Feinstein's lemma for point-to-point channels.
 Our new bound in \eqref{eq:improvedfeinstein} improves on the bound of Blackwell \etal by replacing $\sum_{\theta}\Pbb_{\theta}\biggl[ \log \frac{P_{Y|X,\theta}(Y|X,\theta)}{P_{Y|\theta}(Y|\theta)}\leq \alpha+\delta \biggr]$ therein with $\max_{\theta}\Pbb_{\theta}\biggl[ \log \frac{P_{Y|X,\theta}(Y|X,\theta)}{P_{Y|\theta}(Y|\theta)}\leq \alpha+\delta \biggr]$. 
 %{\color{red}a finite blocklength comparison between the two bounds and between the achievability and LP-based converse?}
 %A $\kappa-\beta$ based hypothesis testing bound has been proposed by Polyanskiy in \cite{polyanskiy2013dispersion}, the exact evaluation of which is cumbersome. For compound channels with channel state information available at the decoder, a generalization of Feinstein's achievability bound has been obtained in \cite{loyka2016general}.
 
Employing this newly derived upper bound on $\ugame$ together with the lower bound on $\lgame$ obtained in \eqref{eq:typelb}, the following theorem then shows
that 
the difference between the upper and lower values of the game vanishes asymptotically as $n \rightarrow \infty$ for $R<\Cl$ and $R>\Cu$.
\begin{theorem}\label{thm:tightregions}
Consider the zero-sum game with  $M=2^{nR}$ and channel conditional distributions as given in \eqref{eq:DMC}. Then, 
\begin{align*}
\limn \ugame = \limn \lgame = 0 \qquad \forall R < \Cl, \\
\limn \ugame = \limn \lgame = 1 \qquad \forall R > \Cu. 
\end{align*}
\end{theorem}
\begin{proof}
Proof is included in Appendix~\ref{app:1}.
\end{proof}
Recall from the introduction that this result is along expected lines. After all, $\Cl$ being the capacity of the compound channel, we must have that for $R<\Cl$, $\ugame$ and hence $\lgame$ approach $0$ for large $n$. Similarly $\Cu$ being the smallest among the capacities of the individual DMCs corresponding to states $\theta \in \Theta$, by the strong converse of channel coding for $R>\Cu$, $\lgame$ and hence $\ugame$ must approach unity. An additional subtlety here is that we have finite blocklength estimates on the ``approximate'' minimax value, which follow from Theorem~\ref{thm:achievability} and Corollary~\ref{cor:typebasedlb}. Consequently, even though for each finite $n$  the game may not admit a saddle point value, for large $n$ and $R<\Cl$ and $R>\Cu$, it admits near-saddle point value (as  is evident from \eqref{eq:gendifference} in the Appendix). 
%Theorem~\ref{thm:tightregions} shows that the LP relaxation based lower bound together with our new upper bound sandwiches tightly the min-max value $\ugame$ and max-min value $\lgame$ in the regions where $R<\Cl$ and $R>\Cu$. 
%Note that the critical points $C$ and $\bar{C}$ have strong interpretations in information theory. In fact,
%Recall from Section~\ref{sec:infobackground} that $\Cl$ in fact  represents the capacity of the compound channel (and mixed channel) and $\Cu$ represents the minimum of the capacities of the individual point-to-point channels, with $\Cl\leq \Cu$.
Moreover, for the class of channels where $C=\Cl=\Cu$, Theorem~\ref{thm:tightregions} shows that a near-saddle point value exists for all $R$ except when $R=C$. Examples where $\Cl=\Cu$ holds are families of discrete memoryless channels where the capacity achieving input distribution is independent of $\theta$~\cite{lapidoth1998reliable}, that includes the binary symmetric channels and binary erasure channels amongst others.

%In the next section, we present a preliminary analysis on the asymptotic tightness of min-max and max-min values when $R<C(\epsilon,q)$ for fixed $\epsilon \in (0,1)$ and $q \in \Pscr(\Theta)$.
\section{Asymptotic Equality of $\ugame$ and $\lgame$ for $\Cl<R<\Cu$}\label{sec:asymtinter}
We now come to the asymptotic equality of the min-max and max-min values of the game when the rate of communication $R$ lies in the range $\Cl <R<\Cu$ and exact characterization of this limiting value. Towards this, we first establish a fundamental lower bound on the max-min value of the game by employing the LP relaxation from Section~\ref{sec:LP}. From the lower bound we then inspire a finite blocklength achievability scheme whose performance asymptotically matches this lower bound 
as $n \rightarrow \infty$ for all but finitely many values of the rate $R.$

\subsection{An LP-Relaxation Based Fundamental Lower Bound on Max-min problem}
We employ the LP relaxation to derive a lower bound on $\ds \liminf_{n\rightarrow \infty} \lgame$. Recall that for each $q \in \Pscr(\Theta)$, the LP relaxation yields that $\lgame \geq \OPT(\DP(q))$ whereby the following limiting bound holds:
$$ \liminf_{n\rightarrow \infty} \lgame \geq \max_{q \in \Pscr(\Theta)} \liminf_{n\rightarrow \infty} \OPT(\DP(q)).$$ Together with the converse obtained in \eqref{eq:typelb}, the LP relaxation thus gives the following limiting bound on the max-min value of the game.
\begin{theorem}\label{thm:LPfundlb}
The LP-relaxation yields the following fundamental lower bound on the max-min value of the game,
% for \cred{any rate $R$} \textbf{is this right?},
\begin{align}
\liminf_{n\rightarrow \infty} \lgame &\geq \max_{q \in \Pscr(\Theta)}\liminf_{n\rightarrow \infty} \OPT(\DP(q)) \geq L(R),\label{eq:LPfundlb} \\
\where \hspace{0.1cm} L(R)=&\max_{q \in \Pscr(\Theta)}\min_{\Theta' \subseteq \Theta} \biggl \lbrace 1-\sum_{\theta \in \Theta'} q(\theta) \mid R \leq  C(\Theta')\biggr \rbrace, \label{eq:LPfundlb1}
\end{align}
for $R \in [0, \Cu]$ and, 
\[C(\Theta')=\max_{P_X}\min_{\theta \in \Theta'}I_{P_{\Xbb}}(\Xbb;\Ybb|\theta), \]
is the capacity of the compound channel formed by $\Theta' \subseteq \Theta.$
\end{theorem}
\begin{proof}
For any $q \in \Pscr(\Theta)$, we have from \eqref{eq:typelb},
\begin{align}
&\lgame \geq \OPT(\DP(q))\stackrel{(a)}{\geq} 
%& \stackrel{(a)}{\geq}\sum_{\theta}q(\theta)\Ibb\biggl\lbrace R \geq I_{\Pbar_n(q)}(\Xbb;\Ybb|\theta)+2\xi+\frac{\log |\Tscr^n|}{n}\biggr \rbrace\non\\& \qquad - \frac{A(\xi)}{n}-\exp(-n \xi)\\
%&=\max_{q(\theta)}\biggl \lbrace\sum_{\theta}q(\theta)\Ibb\biggl\lbrace \hspace{-0.1cm}R \geq I_{\Pbar_n(q)}(\Xbb;\Ybb|\theta)\hspace{-0.1cm}+2\xi+\frac{\log |\Tscr^n|}{n} \hspace{-0.05cm}\biggr\rbrace \non\\& \qquad - \frac{A(\xi)}{n}-\exp(-n \xi)\biggr \rbrace\\
\sum_{\theta \in \Theta'_n}q(\theta)- \frac{A(\xi)}{n}-\exp(-n \xi),\non
%&\geq \min_{\widehat{\Theta}_n \subseteq \Theta} \biggl \lbrace \sum_{\theta \in \widehat{\Theta}_n}q(\theta) | \exists P_{\Xbb} \in \Pscr(\Ascr) \sub \hspace{0.1cm} R< \min_{\theta \in \widehat{\Theta}_n^c}I_{P_{\Xbb}}(\Xbb;\Ybb|\theta)\non \\&+2\xi +\frac{\log |\Tscr^n|}{n} \biggr \rbrace - \frac{A(\xi)}{n}-\exp(-n \xi), \non \\
%&\geq \min_{\widehat{\Theta}_n \subseteq \Theta} \biggl \lbrace \sum_{\theta \in \widehat{\Theta}_n}q(\theta) | \exists P_{\Xbb} \in \Pscr(\Ascr) \sub \hspace{0.1cm} R \geq  I_{P_{\Xbb}}(\Xbb;\Ybb|\theta)+2\xi \non \\&+\frac{\log |\Tscr^n|}{n} \hspace{0.1cm} \forall \theta \in \widehat{\Theta}_n \hspace{0.1cm}{\color{red}\mbox{and}}\hspace{0.1cm} R<I_{P_{\Xbb}}(\Xbb;\Ybb|\theta)\non \\&+2\xi +\frac{\log |\Tscr^n|}{n} \non \\&\hspace{0.1cm}\forall \theta \in \widehat{\Theta}_n^c \biggr \rbrace - \frac{A(\xi)}{n}-\exp(-n \xi), \non
\end{align} 
%\textbf{do we need the and? I think we directly work with the second one since $\Theta{^\prime c}_n$ satisfies it. In fact $R\leq C(\Theta'^c_n) + \xi + \frac{\log |\Tscr^n|}{n}$, so we can directly jump to the last step. Please check and refine the proof.}
where in $(a)$, $\xi>0$ is an arbitrary constant, $A(\xi)$ is independent of $n$ and
$ \Theta'_n:=\biggl \lbrace \theta \in \Theta | R \geq I_{\Pbar_n(q)}(\Xbb;\Ybb|\theta)\hspace{-0.1cm}+2\xi+\frac{\log |\Tscr^n|}{n}\biggr \rbrace $.  Since 
\[R < I_{\Pbar_n(q)}(\Xbb;\Ybb|\theta)\hspace{-0.1cm}+2\xi+\frac{\log |\Tscr^n|}{n} \quad \forall \theta \in (\Theta'_n)^c,\]
we must have 
\[R < \min_{\theta \in (\Theta'_n)^c} I_{\Pbar_n(q)}(\Xbb;\Ybb|\theta)\hspace{-0.1cm}+2\xi+\frac{\log |\Tscr^n|}{n}.\]
It is clear that $\Theta'_n$ satisfies that $R<C((\Theta'_n)^c)+2\xi+\frac{\log |\Tscr^n|}{n}$, whereby for any $q \in \Pscr(\Theta)$,
\begin{align}
\lgame &\geq \min_{\widehat{\Theta} \subseteq \Theta} \biggl \lbrace \sum_{\theta \in \widehat{\Theta}}q(\theta) | R< C( \widehat{\Theta}^c)+2\xi +\frac{\log |\Tscr^n|}{n} \biggr \rbrace \non 
\\&\qquad - \frac{A(\xi)}{n}-\exp(-n \xi), 
\label{eq:ineq2}
\end{align}
Subsequently, taking limit $n \rightarrow \infty$ on both sides of \eqref{eq:ineq2}, noticing that there are only finitely many $\Thetah\subseteq \Theta$, and then maximizing over $q \in \Pscr(\Theta)$ yields that
\begin{align}
\liminf_{n\rightarrow \infty} \lgame
\geq \max_{q \in \Pscr(\Theta)}\min_{\widehat{\Theta} \subseteq \Theta} \biggl \lbrace \sum_{\theta \in \widehat{\Theta}}q(\theta) |  R <  C( \widehat{\Theta}^c) +2\xi  \biggr \rbrace \non.
\end{align}
Our result then follows by taking $\xi \rightarrow 0$. 
\end{proof}
\subsubsection{The quantities $L(R)$ and $U(R)$} 
$L(R)$ defined in \eqref{eq:LPfundlb1} plays a key role in our analysis. This section is devoted to understanding its properties. 
By making a slight change in the definition of $L$, we define the following quantity,
\begin{align}
U(R):=\max_{q \in \Pscr(\Theta)}\min_{\Theta' \subseteq \Theta} \biggl \lbrace 1-\sum_{\theta \in \Theta'} q(\theta) \mid R < C(\Theta')\biggr \rbrace, \label{eq:epsilonR}
\end{align}
for $R \in [0,\Cu].$ 
For convenience, we extend the definition of $L$ and $U$ to $R>\Cu$ by defining,
$$L(R) = U(R):=1 \qquad \forall\ R>\Cu.$$
Observe that we always have $L(R)\leq U(R)$. To see this, note that for any rate $R$, 
\begin{align} \Uu&:=\{\Theta'\subseteq \Theta|R\leq C(\Theta')\}\non \\
& \supseteq \Uo:=\{\Theta'\subseteq \Theta|R< C(\Theta')\} \label{eq:thetaLUR} 
\end{align} 
holds in general, whereby for any $q \in \Pscr(\Theta)$, $\min_{\Theta' \in \Uu}[1-\sum_{\theta \in \Theta'}q(\theta)] \leq \min_{\Theta' \in \Uo}[1-\sum_{\theta \in \Theta'}q(\theta)]$. In particular, for those rates $R$ such that $R=C(\Thetah)$ for some $\Thetah \subseteq \Theta$, it is easy to see that $\Thetah \in \Uu$ and $\Thetah \not \in \Uo$, whereby $\Uu \supset \Uo$, a strict inclusion. In some cases this implies $L(R)<U(R)$ as we see in the example below.
%\textbf{Two points: I think let's change the notation. Use $U(R)$ instead of $U(R)$ (lower and upper). Make this change throughout. Also, we need a proof of why $L(R) \leq U(R)$ in the equation above.}

Employing $\Uo$, $U(R)$ can be equivalently written as,
\begin{align}
U(R)=\max_{q \in \Pscr(\Theta)}\min_{\Theta' \in \Uo}\biggl \lbrace 1-\sum_{\theta \in \Theta'}q(\theta) \biggr \rbrace \label{eq:equivalentUR}.
\end{align}
Interestingly, the following lemma shows that $U(R)$ defined above is in fact equal to the following expression,
\begin{align}
\widetilde{U}(R)= \max_{q \in \Pscr(\Theta)}\inf \biggl \lbrace \epsilon \in [0,1) | R <C_{\epsilon}^{(q)} \biggr\rbrace, \label{eq:epsilonR1}
\end{align} 
where $C_{\epsilon}^{(q)}$ is the 
$\epsilon$-capacity definition of a mixed channel $P_{Y|X}^{(q)}$. Under the average probability  of error criterion, the $\epsilon$-capacity of this channel has been obtained in \cite[Thm 1 and Lem 1(a)]{yagi2014single}  for $\epsilon \in [0,1)$ as
\begin{align}
C_{\epsilon}^{(q)}=\max_{P_{\Xbb}} \sup \biggl \lbrace R \mid \sum_{\theta}q(\theta)\Ibb\{I_{P_{\Xbb}}(\Xbb;\Ybb|\theta)\leq R\}\leq \epsilon \biggr \rbrace, \label{eq:epscapacity}
\end{align}
where $I_{P_{\Xbb}}(\Xbb;\Ybb|\theta)$ is the mutual information between $\Xbb\sim P_{\Xbb} \in \Pscr(\Ascr)$ and $\Ybb \sim \sum_{x \in \Ascr} P_{\Ybb|\Xbb,\theta}(y|x)P_{\Xbb}(x)$.
%is the $\epsilon$- capacity of
%We show the equivalence between $U(R)$ defined in \eqref{eq:epsilonR} and \eqref{eq:epsilonR1} in the following lemma. 
\begin{lemma}
$\widetilde{U}(R)$ in \eqref{eq:epsilonR1} evaluates to $U(R)$ for $R \in [\Cl,\Cu)$, \ie
\begin{align*}
&\widetilde{U}(R)= \max_{q \in \Pscr(\Theta)}\inf \biggl \lbrace \epsilon \in [0,1) | R <C_{\epsilon}^{(q)} \biggr\rbrace=U(R).
 %\label{eq:epsstar}
\end{align*}
%Consequently,
%\begin{align}
%\limn \lgame \geq \max_{q \in \Pscr(\Theta)}\min_{\Theta' \subset \Theta} \biggl \lbrace \sum_{\theta \in \Theta'} q(\theta) \mid R< C(\Theta'^c)\biggr \rbrace. \label{eq:fundamentlb}
%\end{align}
\end{lemma}
\begin{proof}
From the definition of $\epsilon$-capacity in \eqref{eq:epscapacity}, it follows that when $R<C_{\epsilon}^{(q)}$ for a given $q \in \Pscr(\Theta)$, there exists $P_{\Xbb}\in \Pscr(\Ascr)$ such that $\sum_{\theta}q(\theta)\Ibb\{I_{P_{\Xbb}}(\Xbb;\Ybb|\theta)\leq R\}\leq \epsilon$. 
Consequently, 
\begin{align*}
\widetilde{U}(R)&=\max_{q \in \Pscr(\Theta)}\inf \biggl \lbrace \epsilon \in [0,1) \mid \exists P_{\Xbb}\in \Pscr(\Ascr)\hspace{0.1cm} \sub \hspace{0.1cm} \non \\& \qquad  \sum_{\theta}q(\theta)\Ibb\{I_{P_{\Xbb}}(\Xbb;\Ybb|\theta)\leq R\}\leq \epsilon \biggr \rbrace.
\end{align*}
%For any $P_{\Xbb}\in \Pscr(\Ascr)$  let $\Theta'(P_{\Xbb}):=\{\theta \in \Theta| R \geq I_{P_{\Xbb}}(\Xbb;\Ybb|\theta)\}$. 
It is easy to see that for a given $q \in \Pscr(\Theta)$, the inner infimum of $\epsilon$ in the above expression for $\widetilde{U}(R)$ is equivalent to finding the minimum of sums $\sum_{\theta \in \Theta'}q(\theta)$ over strict subsets $\Theta' \subset \Theta$ corresponding to which there exists a $P_{\Xbb}  \in \Pscr(\Ascr)$ such that $ R \geq I_{P_{\Xbb}}(\Xbb;\Ybb|\theta)$ for all  $\theta \in \Theta'$ and $ R<\min_{\theta \in \Theta'^c } I_{P_{\Xbb}}(\Xbb;\Ybb|\theta). $
% \begin{align*}
%\widetilde{U}(R)&=\max_{q \in \Pscr(\Theta)}\inf \biggl \lbrace \epsilon |\exists P_X \in \Pscr(\Ascr) \sub \sum_{\theta \in \Theta'(P_X)}q(\theta ) \leq \epsilon, \biggr \rbrace \\
%&=\max_{q \in \Pscr(\Theta)}\inf_{P_X \in \Pscr(\Ascr)} \biggl \lbrace  \sum_{\theta \in \Theta'(P_X)}q(\theta ) \biggr \rbrace \\
%&=\max_{q \in \Pscr(\Theta)}\min_{\Theta' \subseteq \Theta} \biggl \lbrace \hspace{-0.1cm}\sum_{\theta \in \Theta'}\hspace{-0.1cm}q(\theta)\bigg |\exists P_{\Xbb}\in \Pscr(\Ascr) \hspace{0.1cm}\sub\hspace{0.1cm} R\geq I_{P_{\Xbb}}(\Xbb;\Ybb|\theta) \forall \theta \in \Thetah \mbox{and} R<\hspace{-0.1cm}\min_{\theta \in \Theta'^c } I_{P_{\Xbb}}(\Xbb;\Ybb|\theta)  \biggr \rbrace.
%\end{align*}
Since a minimum is taken over subsets $\Theta'$, the first condition is redundant. 
Consequently, $\widetilde{U}(R)$ can be equivalently written as the following optimization problem,
\begin{align*}
&\max_{q \in \Pscr(\Theta)}\min_{\Theta' \subset \Theta} \biggl \lbrace \hspace{-0.1cm}\sum_{\theta \in \Theta'}\hspace{-0.1cm}q(\theta)\bigg |\exists P_{\Xbb}\in \Pscr(\Ascr) \hspace{0.1cm}\sub\hspace{0.1cm} R<\hspace{-0.1cm}\min_{\theta \in \Theta'^c } I_{P_{\Xbb}}(\Xbb;\Ybb|\theta)  \biggr \rbrace.
\end{align*} 
Now, $R <\min_{\theta \in \Theta'^c } I_{P_{\Xbb}}(\Xbb;\Ybb|\theta) $ for some $P_{\Xbb}$ holds if and only if $R< C(\Theta'^c)$,  which gives the following equivalent expression,
\begin{align*}
&\max_{q \in \Pscr(\Theta)}\min_{\Theta' \subset \Theta} \biggl \lbrace 1-\sum_{\theta \in \Theta'}\hspace{-0.1cm}q(\theta)\bigg |\hspace{0.1cm} R<C(\Theta') \biggr \rbrace.
\end{align*} It is then easy to see that the above expression is equivalent to $U(R)$ for $R\geq \Cl$, which excludes the feasibility of $\Theta'=\Theta$ in \eqref{eq:epsilonR}.
\end{proof}

 The following proposition describes the form of $L$ and $U$ precisely and gives the condition when $L(R)=U(R).$
\begin{proposition}\label{prop:continuousR}
\begin{enumerate}
\item $L$ and $U$ are non-decreasing step functions. $U$ is right-continuous and $L$ is left-continuous.
\item  $L$ and $U$ are discontinuous at the same points. 
\item $L(R)=U(R)$ holds for all rates $R$ where these functions are continuous. In particular, $L(R)=U(R)$ for all rates $R$ such that $R \neq C(\Theta')$ for any $\Theta' \subseteq \Theta$.
\end{enumerate}

\end{proposition}
\begin{proof}
%To prove this, we first note two important properties of the function $U(R)$ : $(a)$ $U(R)$ is non-decreasing in $R$, and $(b)$ $U(R)$ is a step-function.
%To show $(a)$, consider rates $R_1,R_2$ such that $R_1 <R_2$. For any $q \in \Pscr(\Theta)$, let $\Thetah_1=\{\bar{\Theta} \subset \Theta | R_1 <C(\bar{\Theta}^c)\}$ and $\Thetah_2=\{\bar{\Theta} \subset \Theta | R_2 <C(\bar{\Theta}^c)\}$. Since $R_1<R_2$, $\Thetah_2 \subseteq \Thetah_1$ whereby for any $q \in \Pscr(\Theta)$, $$ \min_{\Theta' \in \Thetah_1}\sum_{\theta \in \Theta'}q(\theta)\leq \min_{\Theta' \in \Thetah_2}\sum_{\theta \in \Theta'}q(\theta),$$ which in turn implies that $\epsilon(R_1)\leq \epsilon(R_2)$.
%To see $(b)$,
\textit{1)} We first argue that $U$ is a non-decreasing step function. Note that for rates $R_1,R_2$ such that $R_1\leq R_2$, $\overline{\Upsilon}(R_1)\supseteq \overline{\Upsilon}(R_2)$ (where $\overline{\Upsilon}(R)$ is as defined in \eqref{eq:thetaLUR}), which in turn gives that $U(R_1)\leq U(R_2)$. Thus, $U(R)$ is non-decreasing. Moreover, if for a given $R$, $\Theta_1,\Theta_2 \subseteq \Theta$ are such that $C(\Theta_1)=\max_{\Theta'\subseteq \Theta : R\geq C(\Theta')} C(\Theta')$ and $C(\Theta_2)= \min_{\Theta'\subseteq \Theta: R<C(\Theta')}C(\Theta')$, then $U(R)$ is constant in the interval $[C(\Theta_1),C(\Theta_2))$. Together, we thus have that $U(R)$ is a non-decreasing step function which is continuous in the range $C(\Theta_1)<R<C(\Theta_2)$ and is right-continuous everywhere.
Similarly, it is easy to verify that $L(R)$ is a non-decreasing function. Further, if $\Theta_3,\Theta_4 \subseteq \Theta$ are such that $C(\Theta_3)=\max_{\Theta'\subseteq \Theta : R> C(\Theta')} C(\Theta')$ and $C(\Theta_4)= \min_{\Theta'\subseteq \Theta: R\leq C(\Theta')}C(\Theta')$, $L(R)$ is constant in the interval $(C(\Theta_3),C(\Theta_4)]$. Hence, $L(R)$ is a non-decreasing step function which is continuous in the range $C(\Theta_3)<R<C(\Theta_4)$ and is left-continuous everywhere.

% Consequently, there exists $\bar{\Theta} \in \arg \max_{\Theta' \in \Theta_{U(R)}^c} C(\Theta'^c)$ such that $
% C(\bar{\Theta}^c) \leq R <\min_{\Theta' \in \Theta_{U(R)}}C(\Theta'^c), \non
%$ and
%$U(R)$ evaluates to the same value for all rates $R$ lying in the above range.
%  It is then easy to see that  $U(R)$ is  continuous in the range
%\begin{align}
%C(\bar{\Theta}^c) < R <\min_{\Theta' \in \Theta_{U(R)}}C(\Theta'^c).\label{eq:range}
%\end{align}  However, when  rate is $\bar{R}=\min_{\Theta' \in \Theta_{U(R)}} C(\Theta'^c)$ 
%%(\ie, $R$ takes the upper boundary value of \eqref{eq:range}) 
%with $\Theta^{*}=\arg \min_{\Theta' \in \Theta_{U(R)}} C(\Theta'^c)$, the evaluation of $U(\bar{R})$ considers the set $\Theta_{U(R)} / \{\Theta^{*}\}$, whereby $U(\bar{R})\geq U(R)$ for all $R$ satisfying \eqref{eq:range}. In particular, if $U(\bar{R})>U(R)$, then a jump discontinuity happens at $\bar{R}$ whereby $U(R)$ is a step function.
%

\textit{2 and 3)} We first show that $L(R)=U(R)$ for $R$ where $U$ is continuous. Consider the rates $R$ where $U$  is continuous. Two cases arise -- $(a)$ $R\neq C(\Theta')$ for any $\Theta' \subseteq \Theta$ 
%and $C(\Theta_1)<R<C(\Theta_2)$ 
in which case it is easy to verify that $\Uu=\Uo$, whereby $L(R)=U(R)$, and $(b)$  $R=C(\Theta')$ for some $\Theta'\subseteq \Theta$ and $U(R)$ is continuous. In the second case, the continuity of $U(R)$ ensures that $U(R)=U(R-\delta)$ for some small $\delta>0$. However, $\overline{\Upsilon}(R-\delta)=\Uu$ whereby $U(R-\delta)=L(R)$.
% %It is then easy to see that for rates $R$ where $U(R)$ is continuous, $\Thetah=\{\Theta' \subset \Theta | R\leq C(\Theta'^c)\}$, since $R$ cannot be equal to $\min_{\Theta' \in \Thetah} C(\Theta'^c)$ at which point $U(R)$ is discontinuous. 
 Thus $U(R)=L(R)$ for rates $R$ where $U(R)$ is continuous. It follows that $L$ and $U$ have the same points of discontinuity, thereby proving the second claim.
% Moreover, the points of discontinuity correspond to rates $R$ such that $R=C(\Theta')$ for some $\Theta' \subseteq \Theta$. 
Moreover, $U(R)$ and $L(R)$ in turn coincide for all rates $R$ such that $R \neq C(\Theta')$ for any $\Theta' \subseteq \Theta$.
\end{proof}
To illustrate this relation between $L(R)$ and $U(R)$, we consider the following example of a three state jammer.

\begin{examplee} Consider a three state jammer with $\Theta=\{1,2,3\}$. Let $C(\Theta')$ for $\Theta' \subseteq \Theta$ be such that,
\begin{align*}
\Cl&<C(\{1,2\})<C(\{1,3\})<C(\{2,3\})<\overline{C}=C(\{1\}).
\end{align*}
Evaluating $U(R)$ for this three-state jammer then results in,
\begin{align*}
U(R)=\begin{cases} 
0 &\mbox{if}\hspace{0.1cm} R< \Cl\\
\frac{1}{3} &\mbox{if}\hspace{0.1cm} \Cl\leq R<C(\{1,2\})\\
\frac{1}{2} &\mbox{if}\hspace{0.1cm} C(\{1,2\})\leq R<C(\{1,3\})\\
\frac{1}{2} &\mbox{if}\hspace{0.1cm} C(\{1,3\}) \leq R<C(\{2,3\})\\
\frac{2}{3} &\mbox{if}\hspace{0.1cm} C(\{2,3\})\leq R<\Cu\\
1 &\mbox{if}\hspace{0.1cm} R \geq \Cu\
\end{cases}
\end{align*}
To see this, let us evaluate $U(R)$ when $C(\{1,3\}) \leq R<C(\{2,3\})$. From the definition in \eqref{eq:epsilonR}, we get that, $U(R)=$
\begin{align*}
&\max_{q \in \Pscr(\Theta)} \min \biggl \lbrace q(1), q(2)+q(3),q(1)+q(3),q(1)+q(2) \biggr \rbrace =\frac{1}{2},
\end{align*} which follows by taking $q(1)=\frac{1}{2}$ and $q(2)=q(3)=\frac{1}{4}$.
We now evaluate $L(R)$ for the considered three-state jammer. This results in
\begin{align*}
L(R)=\begin{cases}
0 & \mbox{if}\hspace{0.1cm} R \leq \Cl \\
 \frac{1}{3} &\mbox{if}\hspace{0.1cm} \Cl< R \leq C(\{1,2\})\\
\frac{1}{2} &\mbox{if} \hspace{0.1cm}C(\{1,2\})< R \leq C(\{1,3\})\\
\frac{1}{2} &\mbox{if}\hspace{0.1cm} C(\{1,3\}) < R \leq C(\{2,3\})\\
\frac{2}{3} &\mbox{if}\hspace{0.1cm} C(\{2,3\})< R \leq \Cu \\
1 &\mbox{if}\hspace{0.1cm} R> \Cu
\end{cases}
\end{align*}
Note that $U(R) \neq L(R)$ at points of discontinuity of $U(R)$ and a strict inequality holds. To see this, consider the case when $R=C(\{2,3\})$.  While $U(R)=\frac{2}{3}$, $L(R)$ yields $\frac{1}{2}$. However, at rates $R$ where $U(R)$ is continuous, say $R=C(\{2,3\})-\Delta$, where $\Delta>0$ is  small enough, $U(R)=L(R)=\frac{1}{2}$. Thus, $L(R)$ and $U(R)$ coincide for all rates $R$ where $U(R)$ is continuous. Figure~\ref{fig:epsR} shows the $U(R)$ and Figure~\ref{fig:LR} shows  $L(R)$ for the three-state jammer.
\end{examplee}

\begin{figure}[!h]
\centering
  \includegraphics[scale=0.45,clip=true, trim=0.5in 6.2in 0in 0in]{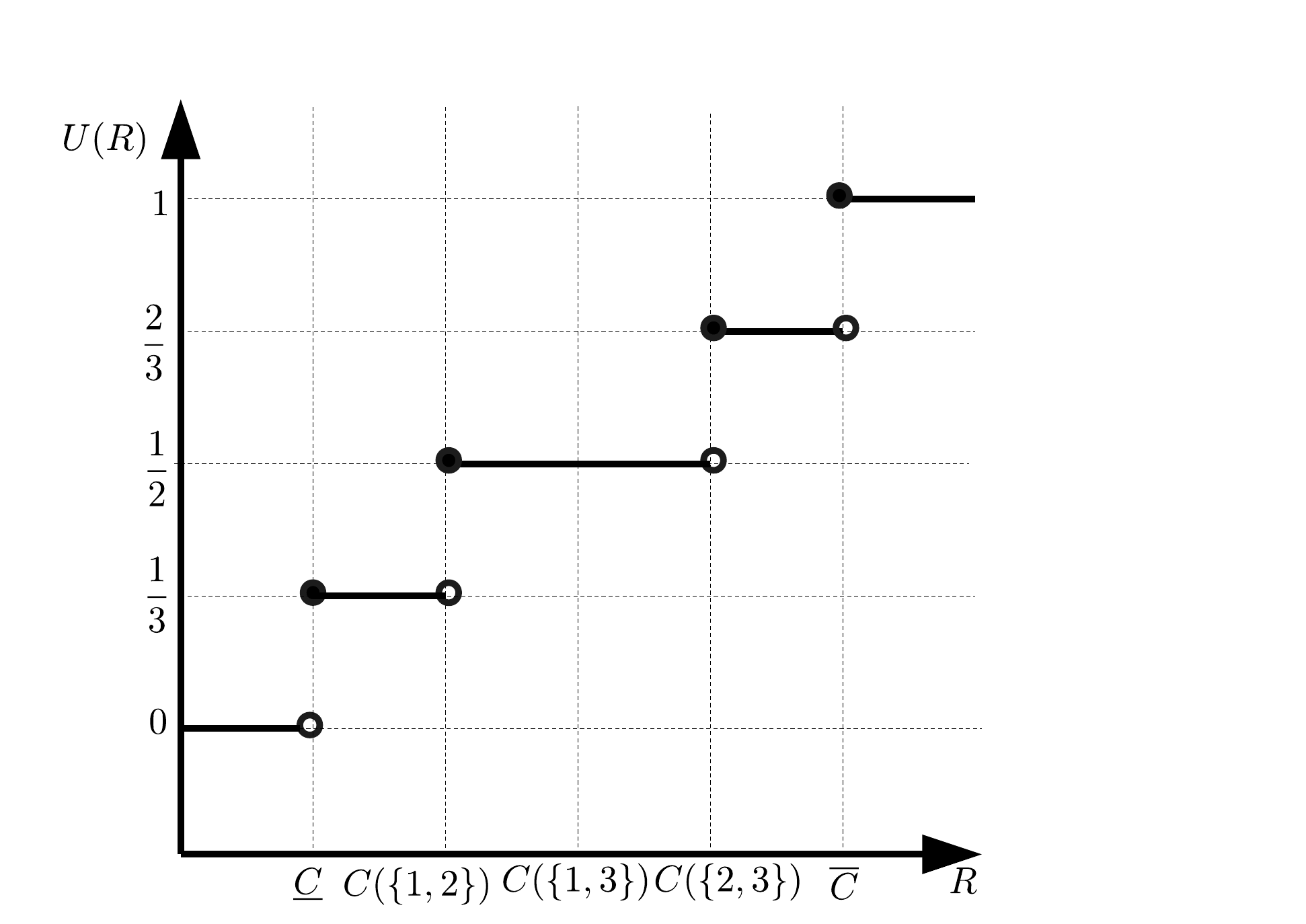}
  \caption{$U(R)$ for three-state jammer. A hollow circle indicates an excluded end-point rate and a filled circle indicates that the end-point rate is included.}
  \label{fig:epsR}
  \includegraphics[scale=0.45,clip=true, trim=0.5in 6.2in 0in 0in]{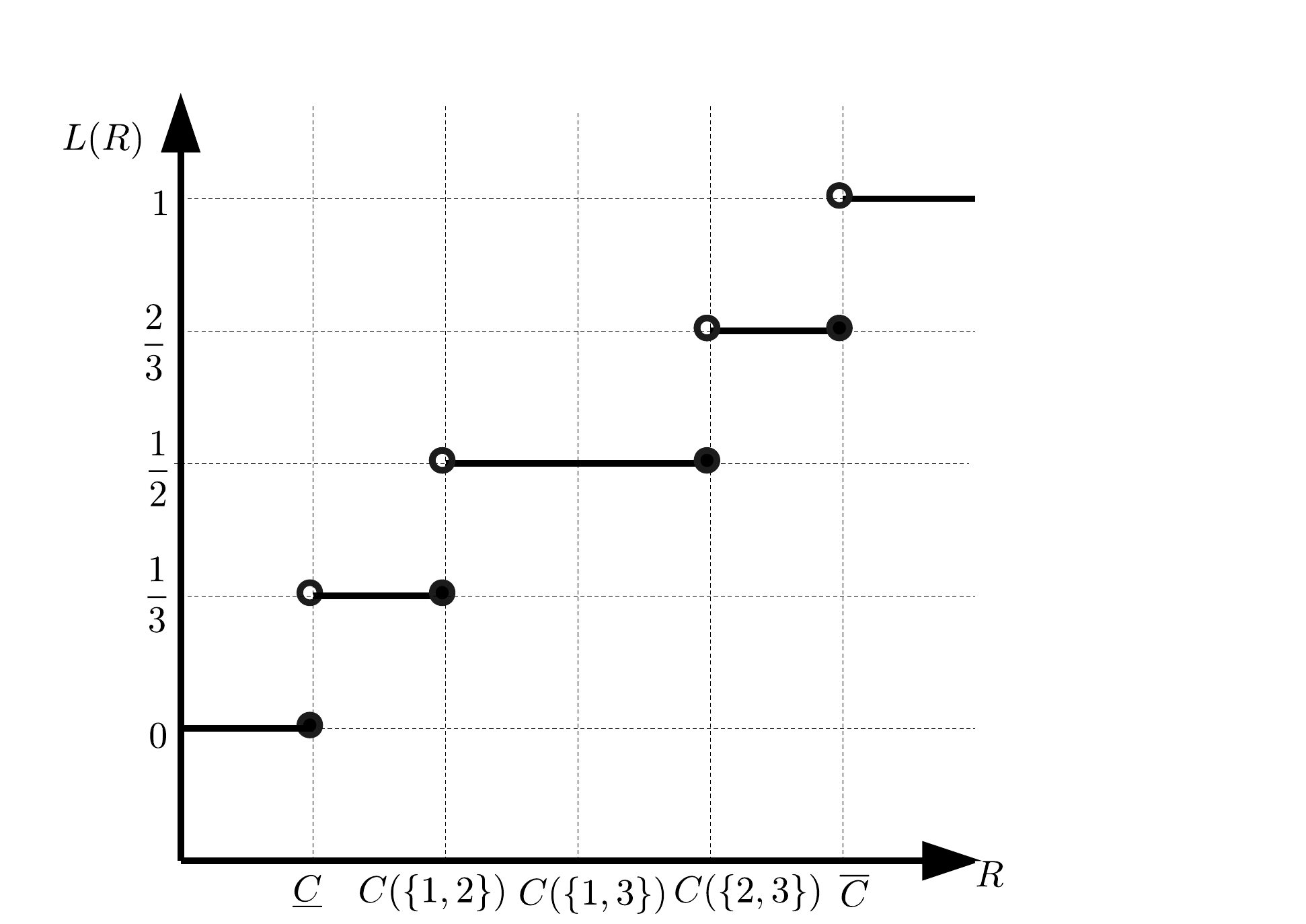}
  \caption{$L(R)$ for three-state jammer. A hollow circle indicates an excluded end-point rate and a filled circle indicates that the end-point rate is included.}
  \label{fig:LR}
\end{figure}
\subsection{Finite Blocklength Achievability Upper Bound on $\ugame$}
In this section, we present a  finite blocklength achievability scheme for a compound channel to obtain an upper bound on the min-max value $\ugame$ of the game when rate $R$ lies in the range $\Cl<R<\Cu$. While Theorem~\ref{thm:LPfundlb} showed a lower bound of $L(R)$ in this range, our scheme will achieve an upper bound of $U(R)$, implying, thanks to Proposition~\ref{prop:continuousR}, that the upper and lower bounds are equal, and equal to $L(R)$, everywhere except at finitely many points.

The proposed scheme has two key features which distinguishes itself from the achievability scheme in Theorem~\ref{thm:achievability} employed for the cases when $R<\Cl$ and $R>\Cu$. Firstly, in contrast to the deterministic encoder employed in Theorem~\ref{thm:achievability}, the encoder in our new scheme avails a local randomization strategy, with its output depending on the outcome of a random experiment. 
% we derived a finite blocklength achievability (see Theorem~\ref{thm:achievability}) by resorting to a deterministic code that yielded a specific construction of codebook and its corresponding decoding sets. 
%our new achievability scheme for the intermediate rate region $\Cl<R<\Cu$ resorts to a deterministic decoder , however the encoder avails a local randomization strategy - a key feature of the achievability scheme.
Precisely stating, for sending each message $s \in \Sscr$, the encoder performs a random experiment independent of $s$, the outcome of which is represented by the random variable $V$ taking values in a known finite space $\Vscr$ and distributed according to a known distribution $P_V \in \Pscr(\Vscr)$. 
% possible set of compound channels, $\{ \{P_{Y|X,\theta}\}_{\theta \in \Theta'}\}_{\Theta' \subseteq \Theta}$. 
 Corresponding to the input message $s \in \Sscr$ and depending on the outcome $V$ of the random experiment, the encoder then outputs $x \in \Xscr$ as $x=f(V,s)$  where $f$ is a function that maps $\Vscr \times \Sscr $ to $\Xscr$. Note however that $V$ is not available to the decoder, whereby the randomization is local to the encoder. Thus, the stochastic encoder in our achievability scheme is taken to be,
\begin{align}
Q_{X|S}(x|s)=\sum_{v \in \Vscr}\Ibb\{x=f(v,s)\}P_V(v).
\label{eq:randomencoder}
\end{align}
 To develop an achievability scheme we employ a split-achievability technique -- the second key feature of the scheme. 
In the scheme we consider, the space $\Vscr$ is taken as $\Uo$, \ie, the collection of subsets $\Theta' \subseteq \Theta$ that each define a compound channel $\{P_{Y|X,\theta}\}_{\theta \in \Theta'}$ with $C(\Theta')>R$, and the  random experiment chooses a compound channel from the above collection.
%We note here that if the outcome $H$ is also available to the decoder, we say that the randomness is shared between the encoder and decoder. Moreover, if for each outcome $k \in \Kscr$, the encoder-decoder pair pulls a deterministic encoding-decoding strategy according to the distribution $P_H$, the resulting code is called a \textit{randomized code} in the sense of Lapidoth \etal \cite{lapidoth1998reliable}.  This in turn is equivalent to resorting to \textit{mixed strategies} in  the language of game theory.
The value of $V$ specifies the codebook that is used for encoding the messages; when $V=v$, the codebook for compound channel $\{P_{Y|X,\theta}\}_{\theta \in v}$ is used to encode messages.  
$f$ is designed such that it encodes both, the value of $V$ and the message. $V$ is encoded into a string of length $n_1$; this string is prefixed before each codeword from the codebook for the compound channel formed by $V$ to get the actual channel input string. 
From the received output, the decoder attempts to decode $V$, and then decodes the message based on the codebook associated with $V$. 
We get an error if $(a)$ the  channel state $\theta \in \Theta$ does not lie in the choice of the compound channel chosen by the random experiment, \ie, $\theta \not \in V$, or $(b)$ the channel state $\theta \in V$, but the decoder fails to correctly decode either the message sent or the choice of the compound channel $V$.

To implement the split-achievability technique outlined above, we split the channel input space as $\Xscr=\Ascr^{n_1} \times \Ascr^{n-n_1}$ and the channel output space as $\Yscr=\Bscr^{n_1} \times \Bscr^{n-n_1}$ where $n_1$ also agreed upon by the decoder. 
%Subseqeuntly, we 
%%adopt two seperate achievability schemes ensuring minimum erroneous transmission - first to code for the outcome $K$ of the random experiment and the second to code for the actual message $s \in \Sscr$ sent given the outcome $K$.
%combine two separate achievability schemes -- the first achievability scheme guarantees the existence of a $n_1$-length code (with encoding and decoding operations occurring in the space of $n_1$-length strings) of size $|\Vscr|$, that codes for the outcome of the random experiment $V$ such that the maximum probability of erroneously decoding $V$ across the compound channel $\{P_{Y|X,\theta}\}_{\theta \in \Theta}$ does not exceed $\lambda \in (0,1]$. For each choice of $v\in \Vscr$, the second achievability scheme  guarantees the existence  of a $(n-n_1)$-length code of size $M$ for the compound channel $\{P_{Y|X,\theta}\}_{\theta \in v}$, that codes for the messages in $\Sscr$ such that maximum probability of erroneously decoding $S$ does not exceed $\lambda$. Note here that $V$ is coded using strings of length $n_1$, which is also agreed upon by the decoder. 
%Subsequently, an error in our new achievability scheme results when there is a failure in correctly decoding either $K$ or the actual message sent.%we split the space of codewords $\Xscr=\Ascr^n$ as $\Xscr=\Xscrh \times \Xscrt$ where
We use the following notation. Let $\Xscrh=\Ascr^{n_1}$ and $\Xscrt=\Ascr^{n-n_1}$.  A generic channel input $X\in \Xscr$ is written as $X=(\Xh,\Xt)$ where  $\Xh \in \Xscrh$, $\Xt \in \Xscrt$ and similarly $x=(\xhat,\xtilde)$ where $\xhat \in \Xscrh,\xtilde \in \Xscrt$.
% Thus, the codeword for $K$ acts a $n_1$-length prefix for the codeword corresponding to the message sent.  The length of the pre-fixed string $n_1 \in \Nbb$ is held constant and is also known to the decoder.
% Let $\Xh \in \Xscrh=\Ascr^{n_1}$, $\Xt \in \Xscrt=\Ascr^{n-n_1}$ with $X=(\Xh,\Xt)$ and $x=(\xhat,\xtilde) \in \Xscr$. 
 Similarly, 
%let the space of channel output is taken to be $\Yscr=\Yscrh \times \Yscrt$  with 
let $\Yscrh=\Bscr^{n_1}$ represent the channel output space corresponding to $\Xscrh$ and  $\Yscrt=\Bscr^{n-n_1}$ represent the channel output space corresponding to $\Xscrt$. Let $Y=(\Yh,\Yt)$ with $\Yh \in \Yscrh$, $\Yt \in \Yscrt$ and $y=(\yhat,\ytilde)$ with  $\yhat\in \Yscrh$, $\ytilde\in \Yscr$. 
With this notation we have that  for each $\theta \in \Theta$,
\begin{align}
P_{Y|X,\theta}(y|x,\theta)&=\prod_{i=1}^n \Pbb_{\Ybb|\Xbb,\theta}(y_i|x_i,\theta)\non\\&=P_{\widehat{Y}|\widehat{X},\theta}(\yhat|\xhat,\theta)P_{\Yt|\Xt,\theta}(\ytilde|\xtilde,\theta),
\end{align}
where 
\begin{align}
P_{\widehat{Y}|\widehat{X},\theta}(\yhat|\xhat,\theta)&=\prod_{i=1}^{n_1} \Pbb_{\Ybb|\Xbb,\theta}(\yhat_i|\xhat_i,\theta) \label{eq:chan1} \\
\aur \quad P_{\Yt|\Xt,\theta}(\ytilde|\xtilde,\theta)&=\prod_{i=n_1+1}^{n} \Pbb_{\Ybb|\Xbb,\theta}(y_i|x_i,\theta)\non \\ &=\prod_{i=1}^{n-n_1} \Pbb_{\Ybb|\Xbb,\theta}(\ytilde_i|\xtilde_i,\theta) \label{eq:chan2} 
\end{align}
% Let the resulting min-max problem with the modified spaces of $\Xscr$ and $\Yscr$ be represented by $\uPnew$ with its optimal value given by $\ugamenew$.

%We now present our  finite blocklength achievability scheme. 
%Towards this, we first define a maximal set.
%\begin{definition}[Maximal Set]
%A set $\Thetah \subset \Theta$ is called a maximal set for rate $R$ if $\Thetah^c \in \Phi(R)$, where $\Phi(R)$ is the collection of minimal sets for $R$. The collection of maximal sets for rate $R$ is then given as,
%\begin{align}
%\Phi'(R)=\{ \Thetah^c \subset \Theta | \Thetah \in \Phi(R)\}.
%\end{align}
%\end{definition} 
\def\fv{\widetilde{f}_v}
\def\gv{\widetilde{g}_v}

The following theorem gives our finite blocklength achievability result.
\begin{theorem}\label{thm:achievinterrange}
For any rate $R$ lying in $\Cl<R<\Cu$, 
%\cred{let the input distribution be $P_X(x)=\sum_{v \in \Vscr}P_V(v)P_{X|V=v}(x)$} \textbf{why do we need this?} where
let $\Vscr=\Uo$ and $V \in \Vscr$ be a random variable  distributed according to
 \begin{align}
 P_V \in \arg \min_{\bar{P}_V \in \Pscr(\Vscr)} \biggl \lbrace  \max_{\theta \in \Theta} \sum_{v \in \Vscr}\bar{P}_V(v)\Ibb\{\theta \not \in v\}\biggr \rbrace. \label{eq:Pv}
 \end{align}
 % some distribution $P_K \in \Pscr(\Theta)$ and
%\begin{align}
%P_K(k)=1-\sum_{\theta \in k^c}q^{*}(\theta),\quad \forall k \in \Phi'(R),
%\end{align}
%where $q^{*} \in \Pscr(\Theta)$ maximizes \eqref{eq:equivalentconverse}. 
% and $P_{X|V=v} \cred{\in} \Pscr(\Xscr)$ for each $v \in \Vscr$.
  Then, for any $\widetilde{\alpha},\widetilde{\delta},\alphah,\deltah>0$, and any $n_1<n$, the following upper bound holds,
%Fix $P_X=\frac{1}{K}\sum_{b \in \Bscr}P_{X|B=b}$, where $P_{X|B=b}$ is any distribution in the space of $\Xscr$ and $\tilde{\alpha},\widehat{\alpha}>0$. Then, there exists a $(M,\epsilon)$ code such that,
\begin{equation}
\ugame \leq U(R)+2\lambda, 
 \label{eq:achievinterrange} 
\end{equation}
where 
\begin{align}
\lambda&=\max\biggl \lbrace \max_{v\in \Vscr} \err_S^{(v)}, \err_V \biggr \rbrace,\quad \mbox{with}\label{eq:lambdavalue}\\
%\end{align}
%%\textbf{not sure about both mins... I think both should be max}
% \begin{align}
\err_V&=\max_{\theta \in \Theta}\widehat{\Pbb}_{\theta}\biggl[ \log \frac{P_{\Yh|\Xh,\theta}(\Yh|\Xh,\theta)}{P_{\Yh|\theta}(\Yh|\theta)}\leq \alphah +\deltah \biggr]\non \\
&+|\Vscr||\Theta|^2\exp(-\alphah)+\exp(-\deltah),\\
%\err_S^{(v)}&=\max_{\theta \in v}\tilde{\Pbb}_{\theta|V=v}\biggl[\log \frac{P_{\Yt|\Xt,\theta}(\Yt|\Xt,\theta)}{P_{\Yt|\theta,V=v}(\Yt)} <\talpha+\widetilde{\delta} \biggr]+ \non \\
%& |v|^2M\exp(-\talpha)+\exp(-\tilde{\delta})\hspace{0.1cm} \mbox{for each} \hspace{0.1cm} v \in \Vscr.
\err_S^{(v)}&=\max_{\theta \in v}\tilde{\Pbb}_{\theta|V=v}\biggl[\log \frac{P_{\Yt|\Xt,\theta}(\Yt_v|\Xt_v,\theta)}{P_{\Yt_v|\theta}(\Yt_v|\theta)} \leq \talpha+\widetilde{\delta} \biggr]+ \non \\
& |v|^2M\exp(-\talpha)+\exp(-\tilde{\delta})\hspace{0.1cm} \mbox{for each} \hspace{0.1cm} v \in \Vscr.
\end{align}
%\begin{align}
%\lambda&=\min \biggl \lbrace \min_{k \in \Kscr}\max_{\theta \in k}\tilde{\Pbb}_{\theta|K=k}\biggl[\log \frac{P_{\Yt|\Xt,\theta}(\Yt|\Xt,\theta)}{P_{\Yt|\theta,K=k}(\Yt)} <\talpha+\tilde{\delta} \biggr]+ \non \\& |\Theta|^2M\exp(-\talpha)+\exp(-\tilde{\delta}), |M'(R)||\Theta|^2\exp(-\alphah)+\non \\& \hspace{-0.2cm} \exp(-\deltah)+\max_{\theta \in \Theta}\widehat{\Pbb}_{\theta}\biggl[ \log \frac{P_{\Yh|\Xh,\theta}(\Yh|\Xh,\theta)}{P_{\Yh|\theta}(\Yh|\theta)}\leq \alphah +\deltah \biggr] \biggr \rbrace. 
%\end{align}
 Here for each $v \in \Vscr$,   $\Xt_v \sim P_{\Xt_v} $ where $P_{\Xt_v}$ is any distribution in $\Pscr(\Xscrt)$,  
$\Yt_v$ denotes the output of $\Xt_v$ under the channel $P_{\Yt|\Xt,\theta}$ defined in \eqref{eq:chan2}, and $P_{\Yt_v|\theta}(\ytilde|\theta) \equiv \sum_{\xtilde \in \Xscrt} P_{\Xt_v}(\xtilde)P_{\Yt|\Xt,\theta}(\ytilde|\xtilde,\theta)$. 
  Likewise, $\Xh \sim P_{\Xh} \in \Pscr(\Xscrh)$ and $\Yh$ is the output of $\Xh$ under the channel $P_{\Yh|\Xh,\theta}$ defined in \eqref{eq:chan1} and $P_{\Yh|\theta}(\yhat|\theta) \equiv \sum_{\xhat \in \Xscrh} P_{\Xh}(\xhat)P_{\Yh|\Xh,\theta}(\yhat|\xhat,\theta)$. Finally, 
%  $ P_{\Yt|\theta,V=v}(\ytilde|\theta)=\sum_{\xtilde} P_{\Xt|V=v}(\xtilde)P_{\Yt|\Xt,\theta}(\ytilde|\xtilde,\theta)$ 
 $\widetilde{\Pbb}_{\theta|V=v}$ denotes the probability with respect to $P_{\Yt|\Xt,\theta}P_{\Xt_v}$ 
% $P_{\Xh} \in \Pscr(\Xscrh)$,  $P_{\Yhat|\theta}(\yhat|\theta)=\sum_{\xhat}P_{\Yhat|\Xhat,\theta}(\yhat|\xhat,\theta)P_{\Xhat}(\xhat)$ 
 and $\widehat{\Pbb}_{\theta}$ is with respect to $P_{\Yh|\Xh,\theta}P_{\Xh}$.
%\textbf{need to mention what $\Xh$ and $\Xt$ are. Moreover, not clear why $\Xt$ is the same for all $v$. I think it can be a different input distribution for each $v$ and the output distribution will also depend on $v$. 
%finally, you need to replace the $K$ with $V$.
%}
\end{theorem}
\begin{proof}
%\begin{enumerate}
To obtain the required bound, we consider the randomized encoder $Q_{X|S}$ as in \eqref{eq:randomencoder} with
%take $
%Q_{X|S,K}(x|s,k)=\Ibb\{x=f_k(s)\},
%$ where
 $f:\Vscr \times \Sscr \rightarrow \Xscrh \times \Xscrt$ and $P_V$ as defined in \eqref{eq:Pv}, 
and a deterministic decoder,
$g:\Yscrh \times \Yscrt \rightarrow \{1,\hdots,M\}$. 
% is such that it partitions the space of $\Bscr^{n_1}$ into $|\Phi'(R)|$ disjoint decoding sets $\{\widehat{\Dscr}_k\}_{k \in \Kscr} \subset \Yscrh$ and the space of $\Bscr^{n-n_1}$ into $M$ disjoint decoding sets $\{\tilde{\Dscr}_{k,j}\}_{j=1}^M \subset \Yscrt$ for each $k \in \Kscr$. The decoder then works as follows. Corresponding to the codeword $x=(\uh_k,\utilde_{k,s})$, the decoder checks the channel output $y=(\yhat,\ytilde)$. If $\yhat \in \widehat{\Dscr}_k$ and $\ytilde \in \tilde{\Dscr}_{k,s}$, the decoder outputs $s$ as the index of the message sent; else it declares an error.
We adopt the following split-achievability strategy. 
For coding the choice of $V$, we consider a deterministic code   of size $|\Vscr|$ in the space of $\Xscrh$ and $\Yscrh$, \ie an encoder $\fhat:\Vscr \rightarrow \Xscrh$ and a decoder $\ghat:\Yscrh \rightarrow \Vscr$, such that the maximum  probability of erroneous transmission of any $v \in \Vscr$ over the compound channel $\{P_{Y|X,\theta}\}_{\theta \in \Theta}$ is $\lambda$. Since $\lambda \geq \err_V$, Theorem~\ref{thm:achievability} applied with $M$ replaced by $|\Vscr|$ guarantees the existence of such a code. Subsequently, for each choice of $v \in \Vscr$, we consider another deterministic code of size $M$ in the space of $\Xscrt$ and $\Yscrt$, \ie, an encoder $\fv:\Sscr \rightarrow \Xscrt$, and a decoder $\gv:\Yscrt \rightarrow \Sscr$,  such that the maximum average probability of erroneous transmission over the compound channel $\{P_{Y|X,\theta}\}_{\theta \in v}$ is $\lambda$. Since $\lambda\geq \err_S^{(v)}$, Theorem~\ref{thm:achievability} applied with $\Theta$ replaced by $v$, then guarantees the existence of such a code. The encoder $f$ and decoder $g$ are then assembled as follows,
\begin{align*}
f(v,s) &= (\fhat(v),\fv(s)) \qquad \forall v\in \Vscr, s\in \Sscr, \\
g(\yhat,\ytilde) & = \widetilde{g}_{\ghat(\yhat)}(\ytilde) \qquad \qquad \forall \yhat \in \Yscrh, \ytilde \in \Yscrt.
\end{align*}
In other words, the encoder $f$ encodes the value of $V$ using $\fhat$ and the value of $S$ using $\fv$ when $V=v.$
The decoder $g$ maps $\Yt$ to a message in $\Sscr$ using a function $\gv$ where $v$ is obtained as $\ghat(\Yhat)$. Thus the decoder first decodes $v$ from $\Yhat$ and then uses the resulting value of $v$ to choose $\gv$, which is then used to decode the message from $\Yt$.
% \cred{Then applying Theorem ... with $\Theta$ replaced by $\Vscr$ ... } \textbf{please fill the details.} 
%\cred{$\lambda\leq \err_S^{(v)}$ \textbf{again this should $\geq$}} then \cred{guarantees} the existence of such a code.

We now show that this scheme achieves the bound in \eqref{eq:achievinterrange}. Recall that $\Vscr=\Uo$, whereby $\Vscr$ comprises of compound channels formed from those subsets $v$ of $\Theta$ such that $R<C(v).$ Thus when $V=v$, error in transmission of $S$ occurs when either the channel state $\theta \notin v$ or when $\theta \in v$ but either $v$ is not correctly decoded or, $v$ is correctly decoded, but $S$ is not correctly decoded. 
Concretely, let $\Pbb_\theta(S \neq \Shat|V=v)$
% $$\cred{\Pbb_{\theta}[\err|V=v]=\frac{1}{M}\sum_{s,y,\shat}P_{Y|X,\theta}(y|f(v,s),\theta)\Ibb\{\shat=g(y)\}\Ibb\{s \neq \shat\}} $$ \textbf{I would not introduce new notation. Why not simply use $\Pbb_\theta(S \neq \Shat|V=v)$?}
 represent the average probability of error in transmitting $S$ given that $V=v$ under the  code $(f,g)$ constructed above when the channel state is $\theta$. Clearly, $\Pbb_\theta(S\neq \Shat|V=v)$ equals the sum
\begin{equation}
\Pbb_\theta(S\neq \Shat, \theta \notin v |V=v)+ \Pbb_\theta(S\neq \Shat, \theta \in v |V=v). \label{eq:master} 
\end{equation}
We upper bound the first term in \eqref{eq:master} by $\I{\theta \notin v}$. For the second term, note that 
\begin{align*}
\Pbb_\theta(S\neq \Shat, \theta \in v |V=v) &= \Pbb_\theta(S\neq \Shat, \theta \in v, \ghat(\Yhat) =v |V=v) \\ 
&+ \Pbb_\theta(S\neq \Shat, \theta \in v, \ghat(\Yhat) \neq v |V=v)
\end{align*}
which are probabilities corresponding to the event that the decoder correctly decodes $V$ but makes an error in decoding $S$, and the event that decoder incorrectly decodes $V$ and $S$. Consequently,
\begin{align*}
&\sum_v P_V(v)\Pbb_{\theta}(S \neq \Shat,\theta \in v|V=v) \leq \\
& \sum_v P_V(v) \biggl[\Pbb_\theta (\gtilde_v(\Yt)\neq S|V=v)+\Pbb_\theta(\ghat(\Yh)\neq V|V=v)\biggr]
\\
&\leq 2 \lambda,
\end{align*}  
where the last inequality follows from the construction of our codes.
 % From the construction of our codes, it is easy to see that the latter error probability is atmost $2\lambda$. 
 Consequently,
\begin{align}
\ugame&\leq \max_{\theta \in \Theta}\biggl \lbrace \sum_{v \in \Vscr} P_V(v) \biggl[\Pbb_{\theta}(S\neq \Shat,\theta \in v|V=v)+ \non\\&\quad \Pbb_{\theta}(S \neq \Shat,\theta \not \in v|V=v) \biggr]\biggr \rbrace \non\\
%&\leq 2\lambda+\max_{\theta\in \Theta}\biggl \lbrace \sum_{v \in \Vscr} P_V(v) \Pbb_{\theta}(S\neq \Shat|V=v) \Ibb\{\theta \not \in v\}\biggr \rbrace \non \\
&\leq 2\lambda+\max_{\theta\in \Theta}\biggl \lbrace \sum_{v \in \Vscr} P_V(v) \Ibb\{\theta \not \in v\}\biggr \rbrace \non. 
%\\
\end{align} 
By the choice of $P_V$ in \eqref{eq:Pv}, the second term in the above bound can be equivalently written as,
%Employing \eqref{eq:Pv} together with the linearity of  $\sum_{v \in \Vscr}\sum_{\theta \in \Theta}q(\theta) \bar{P}_V(v) \Ibb\{\theta \not \in v\}$ in $q \in \Pscr(\Theta)$ then gives that,
\begin{align}
\max_{\theta\in \Theta}&\biggl \lbrace \sum_{v \in \Vscr} P_V(v) \Ibb\{\theta \not \in v\}\biggr \rbrace \label{eq:beautiful}  \\ 
&=\min_{P'_V \in \Pscr(\Vscr)}\max_{\theta\in \Theta}\biggl \lbrace \sum_{v \in \Vscr} P'_V(v) \Ibb\{\theta \not \in v\}\biggr \rbrace \non \\ 
&\stackrel{(a)}{=} \min_{P'_V \in \Pscr(\Vscr)}\max_{q \in \Pscr(\Theta)}\biggl \lbrace \sum_{v \in \Vscr}\sum_{\theta \in \Theta}q(\theta) P'_V(v) \Ibb\{\theta \not \in v\}\biggr \rbrace \non\\
&\stackrel{(b)}{=} \max_{q \in \Pscr(\Theta)}\min_{P'_V \in \Pscr(\Vscr)}\biggl \lbrace \sum_{v \in \Vscr}\sum_{\theta \in \Theta}q(\theta) P'_V(v) \Ibb\{\theta \not \in v\}\biggr \rbrace \non \\
&\stackrel{(c)}{=}\max_{q \in \Pscr(\Theta)}\min_{v \in \Uo}\biggl \lbrace \sum_{\theta \in \Theta}q(\theta) \Ibb\{\theta \not \in v\}\biggr \rbrace \non \\
&= \max_{q \in \Pscr(\Theta)}\min_{v \in \Uo}\biggl \lbrace 1-\sum_{\theta \in v}q(\theta) \biggr \rbrace = U(R), \non
\end{align}
where $(a)$ holds due to the linearity of  the expression in the curly braces in $q \in \Pscr(\Theta)$, 
%which results in \eqref{eq:achievinterrange}. 
%The inequality in $(a)$ follows from \eqref{eq:Pv} \textbf{you also need to explain where the $q$ came from. Also instead of doing this in a series of inequalities, consider that term separately and show that it becomes $U(R).$}  and 
 $(b)$ follows from  von Neumann's minimax theorem and $(c)$ follows again from the linearity of the expression in curly braces in $P_V'$. Comparing the resulting expression with \eqref{eq:equivalentUR}, then yields the required bound.
%by bounding $\Pbb_{\theta}\biggl[\Yh \not \in \widehat{\Dscr}_k \bigcup \Yt \not \in \tilde{\Dscr}_{k,s}|X=(\uh_k,\utilde_{k,s})\biggr]$ $\leq $ $\Pbb_{\theta}\biggl[\Yh \not \in \widehat{\Dscr}_k |X=(\uh_k,\utilde_{k,s})\biggr]+$ $\Pbb_{\theta}\biggl[\Yt \not \in \tilde{\Dscr}_{k,s}|X=(\uh_k,\utilde_{k,s})\biggr]$ and using \eqref{eq:errorforchoiceH} and \eqref{eq:errorformessage}.
%The inequality in $(b)$ follows since there exists atleast one $\bar{k} \in \Kscr$ such that $\theta \in \bar{k}$ which follows from Lemma~\ref{lem:maximalcoverage}.
%\end{enumerate}
\end{proof}

Observe that the bound in Theorem~\ref{thm:achievinterrange} is valid for any $n_1<n$. We now show that $U(R)$ is achievable asymptotically for rates lying in the range $\Cl <R<\Cu$ by letting $n_1$ grow to infinity such that $n_1=o(n)$. 
\begin{theorem} \label{thm:asymptoticub} 
Let $\Cl >0$. For $\Cl<R<\Cu$, the upper bound on the min-max value of the game in \eqref{eq:achievinterrange} yields the following limiting value as $n \rightarrow \infty$, \ie,
\begin{align}
\limsup_{n\rightarrow \infty} \ugame \leq U(R). \label{eq:asymptoticub}
\end{align}
\end{theorem}
\begin{proof}
It suffices to argue that for $\Vscr = \Uo$, we can choose $P_{\Xh},P_{\Xt_v}, \alphah,\widetilde{\alpha},\deltah,\widetilde{\delta}$ and $n_1$ such
that for each $v\in \Vscr$, $\err_S^{(v)} \rightarrow 0$ and $\err_V \rightarrow 0$ as $n \rightarrow \infty$. Towards this, as $n$ increases, we let $n_1$ grow to infinity at a rate of $o(n)$.
%The lower bound on $\limn \lgame$ follows from Theorem~\ref{thm:LPfundlb} together with  Lemma~\ref{lem:equivalentfundlb}. We need to thus evaluate the asymptotics of the upper bound on $\ugame$ derived in Theorem~\ref{thm:achievinterrange}.
%To obtain the required expression from \eqref{eq:achievinterrange}, note that $\lambda \leq \max_{v \in \Vscr}\err_S^{(v)}+\err_V.$ Subsequently, we show that , whereby $\limn \lambda \leq \sum_{v \in \Vscr}\limn \err_S^{(v)}+\limn \err_V =0$ and hence $\limn \ugame \leq U(R)$.

$\err_V$ is an upper bound guaranteed by Theorem~\ref{thm:achievability} for sending $|\Vscr|$ messages over the compound channel $\{P_{Y|X,\theta}\}_{\theta \in \Theta}$ using a code of blocklength $n_1$. Since $n_1 \rightarrow \infty$ as $n \rightarrow \infty$, and $|\Vscr| \leq 2^{|\Theta|}$, a constant, the rate of this code is asymptotically zero. Since $\Cl>0$, arguing as in the proof of Theorem~\ref{thm:tightregions}, we can choose $P_{\Xh},\alphah,\deltah$ such that $\err_V \buildrel{n}\over\rightarrow 0.$

For each $v \in \Vscr$, $\err_S^{(v)}$ is an upper bound guaranteed by Theorem~\ref{thm:achievability} for sending $M$ messages over the compound channel $\{P_{Y|X,\theta}\}_{\theta \in v}$ using a code of blocklength $n-n_1$. Since $n_1=o(n)$, the rate of such a code is asymptotically $R$. Now since $\Vscr = \Uo$, we have from \eqref{eq:thetaLUR} that $R<C(v)$ for each $v \in \Vscr.$ Consequently, arguing again as in the proof of Theorem~\ref{thm:tightregions}, we can choose $\Xt_v,\widetilde{\alpha},\widetilde{\delta}$ such that $\err_S^{(v)} \buildrel{n}\over\rightarrow 0$ for each $v \in \Vscr.$
\end{proof}

That gives the upper bound we were looking for.

\begin{remarkc}
We find it rather pleasing to note how, via von Neumann's minimax theorem, the achievable error term for the compound channel in \eqref{eq:beautiful}, \ie,  $\max_{\theta\in \Theta}\bigl \lbrace \sum_{v \in \Vscr} P_V(v) \Ibb\{\theta \not \in v\}\bigr \rbrace $ becomes exactly the required quantity $U(R)$, which we arrived at from the converse for the \textit{mixed} channel -- a different channel altogether. As outlined in the introduction, our intuitions for this are grounded in the near-convexity of coding problems we discovered in~\cite{jose2016linearIT}. There may be other operational interpretations of this phenomenon that are probably worthy of further investigation. The above scheme, though natural in hindsight, occurred to us only after first deriving the converse expression and reinterpreting it in these `dual' terms. It would be illuminating to find an \textit{ab initio} operational justification for the optimality of this scheme. 
\end{remarkc}

\begin{comment}
The key-ingredients in the asymptotic analysis of our new achievability bound in \eqref{eq:achievinterrange} can be summarized as follows.
\begin{enumerate}
\item The finite blocklength upper bound in \eqref{eq:achievinterrange} is such that $\lambda \rightarrow 0$ as $n\rightarrow \infty$.
\item To ensure that $\lambda$ vanishes in the limit as $n \rightarrow \infty$,
\begin{enumerate}
\item $\Vscr$ is chosen such that $R<C(v)=\max_{P_X}\inf_{\theta \in v}I_{P_X}(\Xbb;\Ybb|\theta)$ for all $v \in \Vscr$,
\item $\lambda$ is upper bounded by the finite blocklength achievability bound in \eqref{eq:improvedfeinstein} corresponding to each of the compound channels in $\Vscr$, and finally
\item for any $v \in \Vscr$, the above obtained finite blocklength bound vanishes in the limit of large blocklengths, since $R<C(v)$ is in fact achievable for the compound channel $\{P_{Y|X,\theta}\}_{\theta \in v}$.
\end{enumerate}
\end{enumerate}
We now move on to analyzing the tightness of the upper and lower bounds derived in this section. 
\end{comment}
\subsection{Asymptotic Tightness of the Min-max and Max-min Values}
We now tie  our story together by consolidating the consequences of Theorem~\ref{thm:asymptoticub}, Theorem~\ref{thm:LPfundlb} and Theorem~\ref{thm:tightregions}. 
% We consider two cases when $R$ lies in $\Cl<R<\Cu$ - $(a)$ $\epsilon(R)$ is continuous at $R$, and $(b)$ $U(R)$ is discontinuous at $R$. 
The following result then holds.
\begin{theorem}\label{thm:erlr}
Let $\Cl>0$. For rates $R\geq 0$ such that $U(R)$ is continuous at $R$, the min-max and max-min values of the game approach $\vartheta(R):=U(R)=L(R)$ as $n \rightarrow \infty$, \ie,
\begin{align*}
\vartheta(R):=\limn \ugame = \limn \lgame = L(R) =U(R).
\end{align*}
In particular, for rates $R$ such that $R \neq C(\Theta')$ for any $\Theta' \subseteq \Theta$, the above equation holds.
\end{theorem}
\begin{proof}
Notice that for $R<\Cl$, $L(R)=U(R)=0.$ Moreover for $R>\Cu$, we defined $L(R)=U(R)=1.$ Thus, Theorem~\ref{thm:tightregions} confirms the above claim for $R<\Cl$ and $R>\Cu.$ 
By combining \eqref{eq:LPfundlb} with \eqref{eq:asymptoticub}, we get that for rates $R$ such that $\Cl<R<\Cu$, the following bound holds in the limit as $n \rightarrow \infty$, \ie,
\begin{align}
U(R)&\geq \limsup_{n\rightarrow\infty} \ugame \geq \liminf_{n\rightarrow \infty} \lgame \geq L(R). \label{eq:asymptight}
\end{align}
In particular, for rates $R$ such that $U(R)$ is continuous at $R$, $L(R)=U(R)$ from Proposition~\ref{prop:continuousR}, whereby the claim holds.
% where note that $\sum_{\theta \in \Thetah^c}q^{*}(\theta)$, for some $\Thetah \in \Phi'(R)$ is in fact equal to $U(R)$ as can be seen from Lemma~\ref{lem:equivalentfundlb}.
\end{proof}

Thus, except for those finitely many rates $R$ where $L$ or $U$ are discontinuous, the min-max and max-min values of the game coincide in the limit as $n \rightarrow \infty$, and they coincide to $\vartheta(R)=U(R)=L(R)$, a value one can explicitly compute.  Theorem~\ref{thm:erlr} also gives a closed-form expression for the $\epsilon$-capacity of a compound channel $\{P_{Y|X,\theta}\}_{\theta \in \Theta}$ under stochastic codes as shown in the following corollary.
\begin{corollary}
For any fixed $\epsilon \in [0,1)$, the $\epsilon$-capacity of the compound channel  $\{P_{Y|X,\theta}\}_{\theta \in \Theta}$ under stochastic codes and average error probability criterion, denoted $C_\epsilon$, is given as,
\begin{align}
C_{\epsilon}:= \sup \biggl \lbrace R | \limn \ugame \leq \epsilon \biggr \rbrace =\sup \biggl \lbrace R | L(R) \leq \epsilon \biggr \rbrace \label{eq:epscapcomp}.
\end{align}
\end{corollary}
\begin{proof}
Denote $\kappa(\epsilon):=\sup \bigl \lbrace R | L(R) \leq \epsilon \bigr \rbrace. $ Since $L$ and $U$ are step functions that are equal everywhere except at points of jump-discontinuity, it follows that $\kappa(\epsilon) = \sup \bigl \lbrace R | U(R) \leq \epsilon \bigr \rbrace.$
We first show that if $R$ is $\epsilon$-achievable for the compound channel, \ie $\limn \ugame \leq \epsilon$, then $R\leq \kappa(\epsilon)$. Towards this, note from \eqref{eq:asymptight} that $R$ is $\epsilon$-achievable implies that $L(R) \leq \epsilon$, which in turn gives that $R \leq \kappa(\epsilon)$.
%\begin{align*}
%R &\leq \sup \{\bar{R} | \limn \bar{\nu}(n;\bar{R}) \leq \epsilon\} \stackrel{(a)}{\leq } \sup \{\bar{R} |  L(\bar{R}) \leq \epsilon\}= C_{\epsilon}^{\CC},
%\end{align*}where $(a)$ follows from \eqref{eq:asymptight}.
Conversely, if $R<\kappa(\epsilon)$, then $U(R) \leq \epsilon$ which shows that $\limn \ugame \leq \epsilon$.
\end{proof}
%\begin{proof}
%From Proposition~\ref{prop:continuousR}, it is clear that when $U(R)$ is continuous at $R$, $L(R)=U(R)$ which in turn evaluates to $\sum_{\theta \in \Thetah}q(\theta)$ for some $\Theta \in \Phi(R)$ as can be seen from Lemma~\ref{lem:equivalentfundlb}. Together with \eqref{eq:asymptoticub}, this then yields the required bound.
%\end{proof}

To the best of our knowledge, ours is the first characterization of the $\epsilon$-capacity of a compound channel under stochastic codes. 

We conclude with some final remarks. For those rates $R$ at which $U(R)$ is discontinuous, Theorem~\ref{thm:erlr} gives that the limiting value of the difference between min-max and max-min values of the game amounts to at most $U(R)-L(R)$. The question then arises if we could modify our achievability scheme so as to yield a limiting value of $L(R)$ at the points of discontinuity. Let $R=C(\Theta')$, $\Theta' \subseteq \Theta$ be a point of discontinuity of $U(R)$. Towards modifying the achievability scheme, one can instead take $\Vscr=\Uu$.  Our modified scheme then gives a limiting value of $L(R)$ at the point of discontinuity $R=C(\Theta')$  \textit{if zero probability of error is achievable} at rate $R=C(\Theta')$ for the compound channel $\{P_{Y|X,\theta}\}_{\theta \in \Theta'}$. Whether this is possible depends on the specific channel laws, not on the capacity alone. This is topic of separate research, which is beyond our present scope.
%Achieving tightness in this case then demands modifying achievability scheme to yield a limiting bound of $L(R)$. 
\section{Conclusion} \label{sec:conclusion}
We considered a game between a team comprising of a finite blocklength encoder and decoder and a finite state jammer where the former team attempts to minimize the probability of error and the jammer attempts to maximize it. The nonclassicality of the information structure renders the team's decision problem nonconvex whereby there may not exist a saddle point value to this game. Despite this, we showed that for all but finitely many rates, an asymptotic saddle-point value exists for this game and derived an exact characterization of this value. Our results demonstrate a deeper relation between compound and mixed channels and provide an new characterization of the $\epsilon$-capacity of a compound channel under stochastic codes.
\section*{Acknowledgment}
The authors would like to thank Himanshu Tyagi from the Indian Institute of Science for helpful discussions on this topic.
\appendices 
\section{Proof of Corollary~\ref{cor:typebasedlb}}\label{app:3}
Crucial to proving Corollary~\ref{cor:typebasedlb} is the following lemma from \cite{yagi2014single}, which is useful in analyzing the asymptotic tightness of the bound in \eqref{eq:converse1}.
\begin{lemma}\label{lemma:converseyagi}
For any fixed $x \in \Ascr^n$, let $P$ denote its type. Let $\delta>0$ be an arbitrary constant and define $C_{\theta|x}^{(n)}(\delta):=$
$$\biggl\lbrace y \in \Bscr^n: \mathrel{\bigg|} \frac{1}{n}\log \frac{P_{Y|X,\theta}(y|x,\theta)}{(PP_{\Ybb|\Xbb,\theta})^{\times n}(y|\theta)} -I_{P}(\Xbb;\Ybb|\theta)\mathrel{\bigg|} \leq \delta\biggr \rbrace,$$
for all $\theta \in \Theta$, where $(PP_{\Ybb|\Xbb,\theta})^{\times n}(y)=\prod_{i=1}^n\sum_{a \in \Ascr} P(a)P_{\Ybb|\Xbb=a,\theta}(y_i|a,\theta)$ and $I_{P}(\Xbb;\Ybb|\Theta)$ is the mutual information evaluated when $X$ has the distribution $P$. Then,
$$
\Pbb[Y \in C_{\theta|x}^{(n)}(\delta)] \geq 1-\frac{A(\delta)}{n},
$$ where $\Pbb$ is with respect to $P_{Y|X=x,\theta}$, $A(\delta)$ is a constant independent of $n, P$ and $\theta$.
\end{lemma}

We now prove Corollary~\ref{cor:typebasedlb}.

\begin{proofarg}[of Corollary~\ref{cor:typebasedlb}]
To obtain the bound in \eqref{eq:typelb}, lower bound \eqref{eq:converse1} by fixing some $q \in \Pscr(\Theta)$, $\gamma=n\xi$, $\xi>0$ and $P_{\Ybar|\theta}$ as in \eqref{eq:pybar}. Note that the choice of $\gamma $ and $P_{\Ybar|\theta}$ are independent of $q$. Let the minimum in the resulting expression be attained by $x_0 \in \Xscr$ with type $\bar{P}_{n}(q)$. Subsequently, apply Lemma~\ref{lemma:converseyagi} as illustrated in \cite[Proof of Theorem 2.2]{jose2017linearitw} to get the following lower bound on \eqref{eq:converse1},
\begin{align}&\sum_{\theta}q(\theta)\Ibb\biggl\lbrace \hspace{-0.1cm}I_{\Pbar_n(q)}(\Xbb;\Ybb|\theta)\hspace{-0.1cm} \leq R -2\xi-\frac{\log |\Tscr^n|}{n}\hspace{-0.05cm}\biggr\rbrace \non\\& \qquad - \frac{A(\xi)}{n}-\exp(-n \xi).
 %\label{eq:indicatorconverse}
\end{align} Subsequently, take supremum over $q \in \Pscr(\Theta)$ to get the required lower bound. 
\end{proofarg}
\section{Proof of Theorem~\ref{thm:achievability}}\label{app:2}
\begin{proofarg}[of Theorem~\ref{thm:achievability}]
To obtain an upper bound on $\ugame$, we design a deterministic code $(f,g)$ for the compound channel $\{P_{Y|X,\theta}\}_{\theta \in \Theta}$. The code will be constructed so that $f(i)=u_i$, $u_i \in \Xscr, i\in \Sscr=\{1,\hdots,M\}$ and
% maps  each input message to an $n$-length string in $ \Xscr$ to get a `codebook' $\{u_1,\hdots,u_M\} \subset \Xscr$ and
  $g:\Yscr \rightarrow \Sscr$ partitions $\Yscr$ into $M$ disjoint decoding sets $\{\Dscr_1,\hdots, \Dscr_M\}$ to yield a $(n,M,\lambda)$ code such that
  $$P_{Y|X,\theta}(Y \in \Dscr_i|X=u_i,\theta)\geq 1-\lambda, \quad \forall \theta \in \Theta, i\in [M], $$
  % Consequently, if the channel output corresponding to $u_i,$ lies in $\Dscr_i$,  then $g$ outputs $i$ as the message sent; else it declares an error.
% An error occurs if $g$ outputs $j$ when the actual message sent was $i, i\neq j$, $i,j \in\{1,\hdots,M\}$. 
%Thus, to derive the claimed upper bound on $\ugame$, we it will suffice to construct
% a codebook and decoding sets such that the average probability of error for any $\theta \in \Theta$ does not exceed $\lambda$,
and 
$$\lambda=\max_{\theta}\Pbb_{\theta}\biggl[ \log \frac{P_{Y|X,\theta}(Y|X,\theta)}{P_{Y}(Y)}\leq \alpha \biggr]+\frac{M|\Theta|^2}{\exp(\alpha)},$$ with $P_{Y}(y)=\frac{1}{|\Theta|}\sum_{\theta}P_{Y|\theta}(y|\theta)=\sum_{\theta}\sum_x P_X(x)P_{Y|X,\theta}(y|x,\theta).$ We will then upper bound $\lambda$ to get the required result. Notice that this strategy ensures a bound on the maximum probability over all messages, as required by the theorem.

Towards this, 
corresponding to each $x \in \Xscr$, we associate $$B_{\theta}(x)=\biggl \lbrace y \in \Yscr : \log \frac{P_{Y|X,\theta}(y|x,\theta)}{P_Y(y)}\geq \alpha\biggr \rbrace.\hspace{0.1cm}$$  Then, for any $x \in \Xscr$ and $\theta \in \Theta$, the following relation holds:
\begin{align}
1&\geq \Pbb_{\theta}[Y \in B_{\theta}(x)|X=x] \non
\\
%& =\sum_{y}P_{Y|X,\theta}(y|x,\theta)\Ibb \biggl \lbrace P_{Y|X,\theta}(y|x,\theta)\geq P_{Y}(y)\exp(\alpha) \biggr \rbrace\non\\
&\geq \sum_{y}P_{Y}(y)\exp(\alpha) \Ibb \biggl \lbrace P_{Y|X,\theta}(y|x,\theta)\geq P_{Y}(y)\exp(\alpha) \biggr \rbrace\non\\
&=\exp(\alpha) \Pbb_{Y}(B_{\theta}(x)). \label{eq:boundalpha1}
\end{align}
%which implies that \begin{align} \Pbb_{Y}(B_{\theta}(x)) \leq \frac{1}{\exp(\alpha)} \forall x,\theta. \end{align}
We now construct our codebook and decoding sets along the lines of a maximal code construction of Feinstein (see \cite{robert1990ash}).
\begin{enumerate}
\item Choose if possible $u_1 \in \Xscr$ such that
$$\min_{\theta}\Pbb_{\theta}[B_{\theta}(u_1)|X=u_1]\geq 1-\lambda.$$
Assign $\Dscr_{\theta,1}=B_{\theta}(u_1)$ as the decoding set of $u_1$ with respect to the channel parameter $\theta$. Subsequently, take $\Dscr_1=\bigcup_{\theta}\Dscr_{\theta,1}$ to be the decoding set corresponding to $u_1$. 
%=\bigcup_{\theta}B_{\theta}(u_1)$, where $\Dscr_{\theta,i}$ represents the decoding region for $u_i$ with respect to the channel parameter $\theta$.
%\item Choose $u_2\neq u_1$ such that
%$$ \inf_{\theta}\Pbb_{\theta}[B_{\theta}(u_2) / \Dscr_1|X=u_2]\geq 1-\lambda.$$ Take $\Dscr_2=\bigcup_{\theta}\Dscr_{\theta,2}=\bigcup_{\theta}[B_{\theta}(u_2)/\Dscr_1]$.
\item Choose $u_k  \in \Xscr \backslash \{u_1,\hdots,u_{k-1}\}$ such that
$$ \min_{\theta}\Pbb_{\theta}[B_{\theta}(u_k) /\bigcup_{j=1}^{k-1} \Dscr_j|X=u_k]\geq 1-\lambda.$$ Take $\Dscr_k=\bigcup_{\theta}\Dscr_{\theta,k}=\bigcup_{\theta}[B_{\theta}(u_k)/\bigcup_{j=1}^{k-1} \Dscr_j]$ to be the decoding region corresponding to $u_k$, where $\Dscr_{\theta,k}$ represents the decoding region for $u_k$ with respect to the channel parameter $\theta$. It is easy to verify that the sets $\Dscr_i$ are mutually disjoint.
\end{enumerate}
Suppose this process stops after $K$ steps, which results in a codebook $\{u_1,\hdots,u_K\}$ with $(\Dscr_1,\hdots,\Dscr_K)$ as the corresponding decoding regions. We now show that $K\geq M$.
 The stopping condition implies that for all $x$, there exists $\theta=\theta_0$ such that,
\begin{align*}
\Pbb_{\theta_0}[B_{\theta_0}(x) /\bigcup_{j=1}^{K} \Dscr_j|X=x]< 1-\lambda,
\end{align*}
which gives that
\begin{align*}
\lambda 
%&< 1-\Pbb_{\theta_0}[B_{\theta_0}(x) /\bigcup_{j=1}^{K} \Dscr_j|X=x]\non\\
&\stackrel{(a)}{<} 1-\Pbb_{\theta_0}[B_{\theta_0}(x)|x]+\Pbb_{\theta_0} [\bigcup_{j=1}^{K} \Dscr_j|X=x]\\
&\leq 1-\Pbb_{\theta_0}[B_{\theta_0}(x)|x]+ \sum_{j=1}^K \Pbb_{\theta_0} [ \Dscr_j|X=x],
%&=1-\Pbb_{\theta_0}[B_{\theta_0}(x)|x]+ \non\\& \sum_{j=1}^K \Pbb_{\theta_0} \biggl[ \bigcup_{\theta}[B_{\theta}(u_j)/\bigcup_{k=1}^{j-1} \Dscr_k]|X=x\biggr].
\end{align*}where the inequality in $(a)$ follows since $P(A\backslash B)\geq P(A)-P(B)$. Now averaging with respect to $x$, we get that
\begin{align*}
\lambda 
%&< 1-\Pbb_{\theta_0}[B_{\theta_0}(X)]\non\\&+ \sum_x P_X(x)\sum_{j=1}^K \Pbb_{\theta_0} \biggl[ \bigcup_{\theta}[B_{\theta}(u_j)/\bigcup_{k=1}^{j-1} \Dscr_k]|X=x\biggr]\non\\
&<\Pbb_{\theta_0}[B_{\theta_0}(X)^c]+\sum_{j=1}^K \Pbb_{Y|\theta_0} [ \Dscr_j]\\
&\stackrel{(a)}{\leq} \Pbb_{\theta_0}[B_{\theta_0}(X)^c]+\sum_{j=1}^K |\Theta|\Pbb_{Y}[ \Dscr_j]\non\\
%&\leq \Pbb_{\theta_0}[B_{\theta_0}(X)^c]+\sum_{j=1}^K |\Theta|\Pbb_{Y} \biggl[ \bigcup_{\bar{\theta}}[B_{\bar{\theta}}(u_j)]\biggr]\\
%&\leq \Pbb_{\theta_0}\biggl[ \log \frac{P_{Y|X,\theta}(Y|X,\theta)}{P_{Y}(Y)}\leq \alpha \biggr]\non\\&\qquad+|\Theta|\sum_{j=1}^K \sum_{\bar{\theta}} \Pbb_{Y} \biggl[ B_{\bar{\theta}}(u_j)]\biggr]\\
&\stackrel{(b)}{\leq}\Pbb_{\theta_0}\biggl[ \log \frac{P_{Y|X,\theta_0}(Y|X,\theta_0)}{P_{Y}(Y)}\leq \alpha \biggr]+\frac{|\Theta|^2K}{\exp(\alpha)} 
%&\leq \Pbb_{\theta_0}\biggl[ \log \frac{P_{Y|X,\theta}(Y|X,\theta)}{P_{Y|\theta}(Y|\theta)}\leq \alpha \biggr]+\frac{K}{\exp(\alpha)}\non\\&\qquad+{\color{red}\sum_{j=1}^K \sum_{\bar{\theta}\neq \theta_0} \Pbb_{Y,\theta_0} \biggl[ B_{\bar{\theta}}(u_j)]\biggr]},
\end{align*}where $(a)$ follows since $P_{Y|\theta}(y|\theta)\leq |\Theta|P_Y(y)$ for any $\theta \in \Theta$. To get to $(b)$, upper bound $\Pbb_{Y}[ \Dscr_j]$ with $\Pbb_{Y} \biggl[ \bigcup_{\theta}[B_{\theta}(u_j)]\biggr]$, which is further bounded by $\sum_{\theta \in \Theta}\Pbb_{Y} [B_{\theta}(u_j)]$. The inequality in $(b)$ then follows by employing \eqref{eq:boundalpha1}.
It can now be verified easily that,
\begin{align*}
&\Pbb_{\theta_0}\biggl[ \log \frac{P_{Y|X,\theta_0}(Y|X,\theta_0)}{P_{Y}(Y)}\leq \alpha \biggr]+M|\Theta|^2\frac{1}{\exp(\alpha)}\\
&\leq \lambda \leq \Pbb_{\theta_0}\biggl[ \log \frac{P_{Y|X,\theta_0}(Y|X,\theta_0)}{P_{Y}(Y)}\leq \alpha \biggr]+K|\Theta|^2\frac{1}{\exp(\alpha)},
\end{align*}
which implies that $K\geq M$.

Analyzing the objective value of $\uP$ over this codebook and decoding sets, we get that
\begin{align*}
&\max_{q \in \Pscr(\Theta)} \sum_{s,y}\frac{1}{M}\sum_{\theta}q(\theta)P_{Y|X,\theta}(y|f(s),\theta)\Ibb\{s\neq g(y)\}\\
&\leq \max_{\theta,s\in \Sscr} \Pbb_{\theta}(Y \in \Dscr_{s}^c|X=u_s)\leq \lambda.
%&=\max_{\theta}\sum_{s=1}^M \frac{1}{M}\biggl[1-\Pbb_{\theta}(Y \in \Dscr_s|u_s)\biggr]\\
%&\leq \max_{\theta}\sum_{s=1}^M \frac{1}{M}\biggl[1-\Pbb_{\theta}(Y \in \Dscr_{s,\theta}|u_s)\biggr]
%\\
%&=\max_{\theta}\biggl[1-\sum_{s=1}^M \frac{1}{M}\Pbb_{\theta}(Y \in \Dscr_{s,\theta}|u_s)\biggr]
\end{align*}
Further, $\lambda$ can be again upper bounded as follows. 
%The steps are inspired from the proof of Lemma 3 in \cite{blackwell1959capacity}.
 \begin{align}
%&\Pbb_{\theta}\biggl[ \log \frac{P_{Y|X,\theta}(Y|X,\theta)}{P_{Y}(Y)}\leq \alpha \biggr]\\
%&=\Pbb_{\theta}\biggl[ \log \frac{P_{Y|X,\theta}(Y|X,\theta)}{P_{Y}(Y)}\leq \alpha, \log \frac{P_{Y|X,\theta}(Y|X,\theta)}{P_{Y|\theta}(Y|\theta)}>\alpha+\delta \biggr]\\&+ \Pbb_{\theta}\biggl[ \log \frac{P_{Y|X,\theta}(Y|X,\theta)}{P_{Y}(Y)}\leq \alpha, \log \frac{P_{Y|X,\theta}(Y|X,\theta)}{P_{Y|\theta}(Y|\theta)}\leq \alpha+\delta \biggr]\\
&\lambda \leq  \max_{\theta \in \Theta}\biggl \lbrace\Pbb_{\theta}\biggl[ \log \frac{P_{Y|X,\theta}(Y|X,\theta)}{P_{Y}(Y)}\leq \alpha, \log \frac{P_{Y|X,\theta}(Y|X,\theta)}{P_{Y|\theta}(Y|\theta)}>\non\\&\alpha+\delta \biggr]+ \Pbb_{\theta}\biggl[  \log \frac{P_{Y|X,\theta}(Y|X,\theta)}{P_{Y|\theta}(Y|\theta)}\leq \alpha+\delta \biggr]\biggr \rbrace+\frac{M|\Theta|^2}{\exp(\alpha)}.\label{eq:achiebound1}
\end{align}
The first term in the RHS can be again upper bounded as,
\begin{align*}
%&\Pbb_{\theta}\biggl[ \log \frac{P_{Y|X,\theta}(Y|X,\theta)}{P_{Y}(Y)}\leq \alpha, \log \frac{P_{Y|X,\theta}(Y|X,\theta)}{P_{Y|\theta}(Y|\theta)}>\alpha+\delta \biggr]\\
&\Pbb_{\theta}\biggl[ \exp(\alpha)P_Y(Y)\geq P_{Y|X,\theta}(Y|X,\theta)>\hspace{-0.1cm}\exp(\alpha+\delta)P_{Y|\theta}(Y|\theta) \biggr]\\
&=\sum_{x,y}P_{Y|X,\theta}(y|x,\theta)P_X(x) \Ibb\biggl \lbrace \exp(\alpha)P_Y(y)\geq P_{Y|X,\theta}(y|x,\theta)\\& \qquad >\exp(\alpha+\delta)P_{Y|\theta}(y|\theta) \biggr \rbrace\\
&\leq \sum_{x,y}P_{Y|X,\theta}(y|x,\theta)P_X(x) \Ibb\biggl \lbrace \frac{\exp(\alpha)}{\exp(\alpha+\delta)}P_Y(y)\hspace{-0.1cm}>\hspace{-0.1cm}P_{Y|\theta}(y|\theta)\hspace{-0.1cm} \biggr \rbrace\\
&=\sum_{y}P_{Y|\theta}(y|\theta) \Ibb\biggl \lbrace \exp(-\delta)P_Y(y)\hspace{-0.1cm}>\hspace{-0.1cm}P_{Y|\theta}(y|\theta) \biggr \rbrace\\
&< \exp(-\delta)\sum_y P_{Y}(y)=\exp(-\delta),
%&\leq \Pbb_{\theta}\biggl[P_Y(Y)\geq \exp(\delta)P_{Y|\theta}(Y|\theta) \biggr]\leq \exp(-\delta),
%&=\sum_{x,y} P_{Y|X,\theta}(y|x,\theta)P_X(x)\Ibb\{P_{Y|\theta}(y|\theta)\leq \exp(-\delta)P_Y(y)\}\\
%&=\sum_y P_{Y|\theta}(y|\theta)\Ibb\{P_{Y|\theta}(y|\theta)\leq \exp(-\delta)P_Y(y)\}\\
%&\leq \sum_y \exp(-\delta)P_Y(y)=\exp(-\delta),
\end{align*}which together with the remaining terms in the RHS of \eqref{eq:achiebound1}, yields the required bound.
\end{proofarg}
\section{Proof of Theorem~\ref{thm:tightregions}}\label{app:1}
For proving Theorem~\ref{thm:tightregions}, we first employ the following lemma from \cite{yagi2014single} to derive an upper bound on \eqref{eq:improvedfeinstein1}, 
which is particularly useful in analysing the asymptotic tightness of the bound.
\begin{lemma}\label{lem:achievabilityyagi}
Let $P_{X}(x)=\prod_{i=1}^n P_{\Xbb}(x_i)$ for all $x \in \Ascr^n$, where $P_{\Xbb} \in \Pscr(\Ascr)$. Then we have that,
\begin{align}
&\Pbb\biggl \lbrace \frac{1}{n} \log \frac{P_{Y|X,\theta}(Y|X,\theta)}{P_{Y|\theta}(Y|\theta)} \leq R+ \rho_{\theta,n} \biggr \rbrace \non\\& \leq \Ibb\{I_{\Pbb_X}(\Xbb;\Ybb|\theta)\leq R+\rho_{\theta,n}+\beta) \}+\frac{F(\beta)}{n} \quad \forall \theta,
\label{eq:ubound1}
\end{align}where $\Pbb$ is with respect to $P_{Y|X,\theta}P_X$ and $F(\beta)$ is a constant independent of $n,\theta$, $\{\rho_{\theta,n}\}\geq 0$ is such that $\lim_{n \rightarrow \infty} \rho_{\theta,n}=0$ and $\beta>0$ denotes an arbitrary constant.
Taking $\limsup$ on both sides yield that
\begin{align}
\limsup_{n\rightarrow \infty} \hspace{0.1cm} & \Pbb\biggl \lbrace \frac{1}{n} \log \frac{P_{Y|X,\theta}(Y|X,\theta)}{P_{Y|\theta}(Y|\theta)} \leq R+ \rho_{\theta,n} \biggr \rbrace \non\\& \leq \Ibb\{I_{\Pbb_X}(\Xbb;\Ybb|\theta)\leq R+\beta) \} \quad \forall \theta.
\end{align}
\end{lemma}
\begin{proofarg}[of Theorem~\ref{thm:tightregions}]
To show the required result, we first obtain an upper bound on the difference between the min-max and max-min values of the game. Towards this,  upper bound \eqref{eq:improvedfeinstein1} using Lemma~\ref{lem:achievabilityyagi} with $P_{X}(x)=\prod_{i=1}^n P_{\Xbb}(x_i)$ for all $x \in \Ascr^n$, where $P_{\Xbb} \in \Pscr(\Ascr)$, $\Delta,\beta>0$ are arbitrary constants and $\{\rho_{n,\theta}\geq 0\}$ is such that $\lim_{n\rightarrow \infty}\rho_{n,\theta}=0$. Together with the lower bound in \eqref{eq:typelb}, we then get that
\begin{align}
&0\leq \ugame-\lgame \non\\
%&\max_{\theta}\Pbb_{\theta}\biggl[ \frac{1}{n} \log \frac{P_{Y|X,\theta}(Y|X,\theta)}{P_{Y|\theta}(Y|\theta)} \leq R+\Delta+\rho_{n,\theta}+\gamma\biggr]
&\leq \max_{\theta \in \Theta}\biggl[ \Ibb \biggl \lbrace I_{P_{\Xbb}}(\Xbb;\Ybb|\theta)\leq R+\Delta+\rho_{n,\theta}+\beta\biggr \rbrace \biggr]+\frac{F(\beta)}{n}\non\\ &+\exp(-n\Delta/2)[|\Theta|^2+1]\non \\&- \biggl ( \max_{q \in \Pscr(\Theta)}\sum_{\theta}\hspace{-0.1cm}q(\theta)\Ibb\biggl\lbrace \hspace{-0.1cm}I_{\Pbar_n(q)}(\Xbb;\Ybb|\theta)\hspace{-0.1cm} \leq R -2\xi-\frac{\log |\Tscr^n|}{n}\hspace{-0.05cm}\biggr\rbrace \non\\&- \frac{A(\xi)}{n}-\exp(-n \xi) \biggr),\label{eq:gendifference}
\end{align}where $\xi>0$ is an arbitrary constant and $\Pbar_n(q)$ is the type of $x \in \Ascr^n$ that minimizes \eqref{eq:Pnq} with $P_{\Ybar|\theta}(y|\theta)$ as 
%$$\sum_{\theta}q(\theta) \sum_y P_{Y|X,\theta}(y|x,\theta) \Ibb\biggl \lbrace \log \frac{P_{Y|X,\theta}(y|x,\theta)}{P_{\Ybar|\theta}(y|\theta)} \leq nR-n\xi\biggr \rbrace$$ with $P_{\Ybar|\theta}(y|\theta)$ as 
in \eqref{eq:pybar}.
We show that the difference bound in \eqref{eq:gendifference} vanishes asymptotically for the two cases when $R<\Cl$ and $R>\Cu$.
\\
Case 1: $R>\Cu$.\\
In this case, consider $R=\Cu+\epsilon$, where $\epsilon>0$. Let $\bar{\theta} \in \arg \min_{\theta}\bigl[\max_{P_{\Xbb}}I_{P_{\Xbb}}(\Xbb;\Ybb|\theta)\bigr]$ whereby $R = \max_{P_{\Xbb}}I_{P_{\Xbb}}(\Xbb;\Ybb|\bar{\theta})+\epsilon$. Consequently, upper bounding $\max_{\theta}\bigl[ \Ibb \bigl \lbrace I_{P_{\Xbb}}(\Xbb;\Ybb|\theta)\leq R+\Delta+\rho_{n,\theta}+\beta\bigr \rbrace \bigr]$ with 1 and choosing $q(\theta)=\tilde{q}(\theta)=\Ibb\{\theta=\bar{\theta}\}$, $\xi=\frac{1}{4}\epsilon$,  \eqref{eq:gendifference} can be upper bounded as,
 \begin{align}
% &1+\frac{F(\beta)}{n}+\exp(-n\Delta/2)[|\Theta|^2+1]+\frac{A(\xi)}{n}+\exp(-n \frac{1}{4}\epsilon)\non \\&- \Ibb\biggl\lbrace \hspace{-0.1cm}I_{\Pbar_n(\tilde{q})}(\Xbb;\Ybb|\bar{\theta)}\hspace{-0.1cm} \leq \bar{C}+\frac{1}{2}\epsilon-\frac{\log |\Tscr^n|}{n}\hspace{-0.05cm}\biggr\rbrace\non\\
 & 1+\frac{F(\beta)}{n}+\exp(-n\Delta/2)[|\Theta|^2+1]+\frac{A(\xi)}{n}+\exp(-n \frac{\epsilon}{4})\non \\&- \Ibb\biggl\lbrace \hspace{-0.1cm}I_{\Pbar_n(\tilde{q})}(\Xbb;\Ybb|\bar{\theta)}\hspace{-0.1cm} \leq \max_{P_{\Xbb}} I_{P_{\Xbb}}(\Xbb;\Ybb|\bar{\theta})+\frac{1}{2}\epsilon-\frac{\log |\Tscr^n|}{n}\hspace{-0.05cm}\biggr\rbrace\non
 \end{align}
% \begin{align}
% &\leq 1+\frac{F(\beta)}{n}+\exp(-n\Delta/2)[|\Theta|^2+1]+\frac{A(\xi)}{n}+\exp(-n \frac{\epsilon}{4})\non \\&- \Ibb\biggl\lbrace \hspace{-0.1cm}I_{\Pbar_n(\tilde{q})}(\Xbb;\Ybb|\bar{\theta)}\hspace{-0.1cm} \leq  I_{\Pbar_n(\tilde{q})}(\Xbb;\Ybb|\bar{\theta})+\frac{1}{2}\epsilon-\frac{\log |\Tscr^n|}{n}\hspace{-0.05cm}\biggr\rbrace\non,
% %\\
%% &=1+\frac{F(\beta)}{n}+\exp(-n\Delta/2)[|\Theta|^2+1]+\frac{A(\xi)}{n}+\exp(-n \frac{\epsilon}{4})\non \\&- \Ibb\biggl\lbrace \frac{1}{2}\epsilon\geq \frac{\log |\Tscr^n|}{n}\hspace{-0.05cm}\biggr\rbrace\non.
%\end{align} Here the negative term goes to 1 in the limit since  $\frac{\log |\Tscr^n|}{n}$ goes to zero in the limit as seen from \eqref{eq:typecounting}. Moreover, all other positive terms (except for the number 1) vanish as $n\rightarrow \infty$. Hence, the difference in \eqref{eq:gendifference} goes to zero asymptotically.\\
Here the negative term (and thus the lower bound on $\lgame$) goes to 1  in the limit since $I_{\Pbar_n(\tilde{q})}(\Xbb;\Ybb|\bar{\theta)} \leq \max_{P_{\Xbb}} I_{P_{\Xbb}}(\Xbb;\Ybb|\bar{\theta})$ and $\frac{\log |\Tscr^n|}{n}$ goes to zero in the limit as seen from \eqref{eq:typecounting}. Moreover, all other positive terms vanish as $n\rightarrow \infty$. Hence, the difference in \eqref{eq:gendifference} goes to zero asymptotically.\\
Case 2: $R<\Cl$.\\
In this case, we take $R=\Cl-\epsilon$, where $0<\epsilon<\Cl$. Let $P^{*}_X\in \arg \max_{P_{\Xbb}} \bigl[\min_{\theta} I_{P_{\Xbb}}(\Xbb;\Ybb|\theta)\bigr]$ and fix $\Delta=\beta=\frac{\epsilon}{4}$, $\rho_{n,\theta}=\frac{\epsilon}{n}$ in \eqref{eq:gendifference}. Further, the maximum over $q(\theta)$ term in \eqref{eq:gendifference} can be upper bounded by  0 to yield the following upper bound on \eqref{eq:gendifference} for any $\xi>0$,
\begin{align*}
&\max_{\theta}\biggl[ \Ibb \biggl \lbrace I_{P_{X^{*}}}(\Xbb;\Ybb|\theta)\leq \Cl- \frac{\epsilon}{2}+\frac{\epsilon}{n}\biggr \rbrace \biggr]+K(n)\non\\
%\frac{F(\frac{\epsilon}{4})}{n}\non\\ &+\exp(-n\frac{\epsilon}{8})[|\Theta|^2+1]+\frac{A(\xi)}{n}+\exp(-n \xi)\non \\
%\end{align*}
%\begin{align*}
&\stackrel{(a)}{=} \Ibb \biggl \lbrace I_{P_{X^{*}}}(\Xbb;\Ybb|\theta^{*})\leq \min_{\theta}I_{P_{X^{*}}}(\Xbb;\Ybb|\theta)- \frac{\epsilon}{2}+\frac{\epsilon}{n}\biggr \rbrace +K(n) \non\\ 
%&+\frac{F(\frac{\epsilon}{4})}{n}+\exp(-n\frac{\epsilon}{8})[|\Theta|^2+1]+\frac{A(\xi)}{n}+\exp(-n \xi),\non
%&\leq \Ibb \biggl \lbrace I_{P_{X^{*}}}(\Xbb;\Ybb|\theta^{*})\leq I_{P_{X^{*}}}(\Xbb;\Ybb|\theta^{*})- \frac{\epsilon}{2}+\frac{\epsilon}{n}\biggr \rbrace \non\\ &+\frac{F(\frac{\epsilon}{4})}{n}+\exp(-n\frac{\epsilon}{8})[|\Theta|^2+1]+\frac{A(\xi)}{n}+\exp(-n \xi)\non
%\end{align*}
%\begin{align*}
&\leq \Ibb \biggl \lbrace \frac{\epsilon}{2}\leq \frac{\epsilon}{n}\biggr \rbrace +K(n)\non,
\end{align*} where $K(n)=\frac{F(\frac{\epsilon}{4})}{n}+\exp(-n\frac{\epsilon}{8})[|\Theta|^2+1]+\frac{A(\xi)}{n}+\exp(-n \xi), $ $\theta^{*}\in \arg \max_\theta\Ibb \biggl \lbrace I_{P_{X^{*}}}(\Xbb;\Ybb|\theta)\leq \Cl- \frac{\epsilon}{2}+\frac{\epsilon}{n}\biggr \rbrace $ in $(a)$. It is clear that the above bound vanishes as $n\rightarrow \infty$. 
\end{proofarg}
\bibliographystyle{IEEEtran}
%\bibliography{ref,ref1,apsbib}
\bibliography{ref}
\end{document}